\documentclass[journal,10pt]{IEEEtran}
\usepackage{cite}
\usepackage{subfigure}
\usepackage{graphicx}
\usepackage{amsmath}
\usepackage{txfonts}
\usepackage{times}
\usepackage{latexsym}
\usepackage{graphicx}
\usepackage{bm}
\usepackage{amssymb}
\usepackage[center]{caption2}
\usepackage{stfloats}
\usepackage{cases}
\usepackage{array}
\usepackage{setspace}
\usepackage{fancyhdr}

\usepackage{algorithm}
\usepackage{algorithmic}
\usepackage{CJK}
\usepackage{xcolor}

\graphicspath{{figures/}}

\allowdisplaybreaks[4]

\newcaptionstyle{mystyle1}{%
  \centering TABLE  \captiontext \par}
\captionstyle{mystyle1}
\newcaptionstyle{mystyle2}{%
  \captionlabel $.$ \: \captiontext \par}
\captionstyle{mystyle2}
\newcaptionstyle{mystyle3}{%
  \captionlabel $.$ \: \captiontext \par}
\captionstyle{mystyle3}

\renewcommand{\QED}{\QEDopen}

\newtheorem{theorem}{Theorem}
\newtheorem{corollary}{Corollary}

\begin{document}

\title{Decentralized Asynchronous Coded Caching Design and Performance Analysis in Fog Radio Access Networks}


\author{Yanxiang~Jiang,~\IEEEmembership{Senior~Member,~IEEE}, 
Wenlong~Huang, Mehdi~Bennis,~\IEEEmembership{Senior~Member,~IEEE}, and~Fu-Chun~Zheng,~\IEEEmembership{Senior~Member,~IEEE} 
\thanks{Manuscript received May 25, 2018, revised December 2, 2018, and \today.}
\thanks{This work was supported in part by
the Natural Science Foundation of Jiangsu Province under grant
BK20181264,
the Research Fund of the State Key Laboratory of
Integrated Services Networks (Xidian University) under grant ISN19-10,
the Research Fund of the Key Laboratory of Wireless Sensor Network $\&$ Communication (Shanghai Institute of Microsystem and Information Technology, Chinese Academy of Sciences) under grant 2017002,
the National Basic Research Program of China
(973 Program) under grant 2012CB316004,
a Shenzhen Municipality/HITSZ Start-Up Grant entitled ``Energy-Efficient Low-Latency Wireless Networks",
and the U.K. Engineering and Physical Sciences Research Council under Grant EP/K040685/2.
Part of this work has been presented at 2018 IEEE 88th Vehicular Technology Conference (VTC2018-Fall), Chicago, USA, August 2018.}
\thanks{Y. Jiang is with the National Mobile Communications Research Laboratory, Southeast University, Nanjing 210096, China,
the State Key Laboratory of Integrated Services Networks, Xidian University, Xi'an 710071, China, and the Key Laboratory of Wireless Sensor Network $\&$ Communication, Shanghai Institute of Microsystem and Information Technology,
Chinese Academy of Sciences, 865 Changning Road, Shanghai 200050, China (e-mail: yxjiang@seu.edu.cn).}
\thanks{W. Huang is with the National Mobile Communications Research Laboratory, Southeast University, Nanjing 210096, China (e-mail: ahhwl0514@163.com).}
\thanks{M. Bennis is with the Centre for Wireless Communications, University of Oulu, Oulu 90014, Finland (e-mail: mehdi.bennis@oulu.fi).}
\thanks{F. Zheng is with the School of Electronic and Information Engineering, Harbin Institute of Technology, Shenzhen 518055, China,
and the National Mobile Communications Research Laboratory, Southeast University, Nanjing 210096, China. (e-mail: fzheng@ieee.org).}
}

\maketitle

\begin{abstract}
In this paper, we investigate the problem of asynchronous coded caching in fog radio access networks (F-RANs). To minimize the fronthaul load, the encoding set collapsing rule and  encoding set partition method are proposed to establish the relationship between the coded-multicasting contents for asynchronous and synchronous coded caching. Furthermore, a decentralized asynchronous coded caching scheme is proposed, which provides asynchronous and synchronous transmission methods for different delay requirements.
The closed-form expression of the fronthaul load is established for the special case where the number of requests during each time slot is fixed,
and the upper and lower bounds of the fronthaul load are given for the general case where the number of requests during each time slot is random. The simulation results show that our proposed scheme can create considerable coded-multicasting opportunities in asynchronous request scenarios.
\end{abstract}


\section{Introduction}

With the rapid proliferation of smart devices and mobile application services, wireless networks have been suffering an unprecedented data traffic pressure in recent years, especially at peak-traffic moments.
Fog radio access networks (F-RANs), which can effectively reduce the data traffic pressure by placing popular contents closer to users, have been receiving significant attention from both industry and academia \cite{Bennis1, Bennis2}.
In F-RANs, fog access points (F-APs) are distributed at the edges and connected to the cloud server through fronthaul links.
F-APs can use edge computing and caching resources to provide users better quality of experience \cite{Zhang}.
Meanwhile, since a few popular content resources account for most of the traffic load,  edge caching has become instrumental in content delivery \cite{Bastug, Wang, add10}.
Moreover, coded caching was firstly proposed in \cite{Maddah-Ali1} and \cite{Maddah-Ali2} by encoding the delivered contents to further reduce the data traffic pressure.

The main idea of coded caching is that the contents stored in the caches  can be used to create coded-multicasting opportunities, such that a single coded-multicasting content transmitted by the cloud server can be useful to a large number of users simultaneously even though  they are not requesting the same content. In \cite{Maddah-Ali1}, Maddah-Ali and Niesen  proposed a centralized coded caching scheme, in which the centrally coordinated placement phase needs the knowledge of the number of active users in the delivery phase. A decentralized coded caching scheme was further proposed in \cite{Maddah-Ali2}, which achieves order-optimal memory-load tradeoff in the asymptotic regime with infinite file size.
However, when the file size is sub-exponential with respect to the number of users, the Maddah-Ali-Niesen's decentralized scheme in \cite{Maddah-Ali2} achieves at most a multiplicative gain of two over the conventional uncoded caching scheme\cite{shan}.
Focusing on the finite file size regime, the authors in \cite{Jin} proposed a decentralized random coded caching scheme and a partially decentralized sequential coded caching scheme, which outperform the Maddah-Ali-Niesen's decentralized scheme when the file size is not very large.

Furthermore,
the authors in \cite{Niesen1} presented  a strategy which partitions the file library into subsets of approximately uniform request probability and applies the strategy in \cite{Maddah-Ali2} to each subset.
In \cite{Hachem}, the authors considered the case where the entire content is divided into multiple different levels based on popularity and an information-theoretic outer bound was  developed.
A scheme consisting of a random popularity-based caching policy and chromatic-number index coding delivery was proposed in \cite{Ji1}, which was proven to be order optimal in terms of average rate.
In \cite{JZhang1}, the authors considered an arbitrary popularity distribution and  derived a new information-theoretical lower bound on the expected transmission rate of any coded caching schemes.
In \cite{Ji2},  the analysis in \cite{Ji1} was extended to the case where each user requests multiple files
 and an order-optimal delivery scheme was provided  based on local graph coloring, which achieved the order gain compared with the repeated application of the scheme in \cite{Maddah-Ali1}.
In \cite{Ji3}, the analysis was further extended to the case with distinct cache sizes and demand distributions, and  a
novel polynomial-time algorithm based on greedy local graph coloring was proposed, which can
recover a significant part of the multiplicative caching gain with the same content packetization.
In \cite{JZhang} and \cite{ Amiri}, more complex heterogeneous network settings for file sizes and cache capacities were studied, respectively.
In \cite{Pedarsani}, an online coded caching scheme  was proposed, which  updates cache contents by evicting some old file parts and replacing them with randomly chosen parts of the newly delivered files.
{
In \cite{add1}, the authors reformulated the centralized coded caching problem as designing a corresponding placement delivery array
and proposed two new schemes from this perspective which can significantly reduce the rate compared with the uncoded caching schemes.
In \cite{add2}, the authors 
viewed the centralized coded caching problem in a hypergraph perspective
and showed that designing a feasible placement delivery array is equivalent to constructing a linear hypergraph in extreme graph theory.
In \cite{add3}, the authors proposed coded caching schemes based on combinatorial structures called resolvable designs, which can be obtained in a natural manner from linear block codes whose generator matrices possess certain rank properties,
and obtained several schemes with subpacketization levels substantially lower than the basic scheme at the cost of an increased rate.
In \cite{add4}, improved lower bounds on the required rate for the coded
caching problem was developed and it was demonstrated that the computation of this lower bound can be posed as a combinatorial labeling problem on a directed tree.
In \cite{add5}, a coded prefetching and the corresponding delivery strategy was proposed, which relies on a combination of rank metric codes and maximum distance separable (MDS) codes in a non-binary finite field.
In \cite{add6}, a novel centralized coded caching scheme was proposed that approaches the rate-memory region achieved by the scheme in \cite{add5} as the number of users in the system increases, which only requires a finite field of $2^2$.
Moreover, instead of relying on the existence of some valid code, an explicit combinatorial construction of the caching scheme was provided.
In \cite{add7}, the authors proposed a connection between the uncoded prefetching scheme proposed by Maddah Ali and Niesen and the coded prefetching scheme in \cite{add5}.
The new general coding scheme was then presented and analyzed rigorously, which yields a new inner bound to the memory-rate tradeoff for the caching problem.
}

All the above schemes in \cite{Maddah-Ali1, Maddah-Ali2, Niesen1, Hachem, Ji1, Ji2, Ji3, Pedarsani, Jin, Amiri, JZhang1, shan, JZhang, add1, add2, add3, add4, add5, add6, add7} considered the coded caching problem for the case in which user requests are  synchronous, i.e., synchronous coded caching. However,
user requests for contents are typically asynchronous in reality \cite{Maddah-Ali3}.
The asynchronous request case was first mentioned in \cite{Maddah-Ali2}, and the authors applied their proposed decentralized synchronous coded caching scheme to an asynchronous request scenario in a simple way.
{In \cite{Niesen2}, the delay sensitive coded caching problem was first studied and the situation whereby each asynchronous request has a specific deadline was considered. Then, a computationally efficient caching scheme  that exploits coded-multicasting  opportunities was developed  subject to the delivery-delay constraint.}
In \cite{Ghasemi},
the authors proposed a linear programming  formulation for the offline case that the server knows the arrival time before starting transmission. As for the online case that user requests are revealed to the server over time, they considered the situation that users do not have deadlines but wish to minimize the overall completion time.
{In \cite{add8}, a centralized coded joint pushing and caching (C-JPC) method with asynchronous user requests was proposed to minimize the network traffic by jointly determining when and which data packets are to be pushed and whether they should be cached.
Optimal offline and online C-JPC policies for noncausal and causal request delay information were obtained by solving optimization problems.
Fountain coded caching (FCC) and generalized coded caching (GCC) methods were further proposed to give sub-optimal policies with low complexity.
The authors analyzed the bounds on the optimal traffic volume and proved that the FCC and GCC methods achieve optimal or near-optimal traffic volumes in some special cases.}

Motivated by the aforementioned discussions, it is important to study the coded caching problem when  user requests are  asynchronous, i.e.,   asynchronous coded caching.
In view of this, we consider the online case with a given maximum request delay to reduce the worst-case load of the fronthaul links in F-RANs.
Our main contributions are summarized below.

\begin{enumerate}
\item

 We propose an encoding set collapsing rule to establish the relationship between the coded-multicasting contents in asynchronous and synchronous coded caching. Furthermore, we propose an encoding set partition method, which can create {considerable coded-multicasting opportunities} while the delay of each user is no more than the given maximum request delay.

\item

We propose a decentralized asynchronous coded caching scheme, which can  exploit the created coded-multicasting opportunities effectively.
Our proposed scheme is applicable for various asynchronous request scenarios by  providing asynchronous and synchronous transmission methods, which can be chosen according to different  delay requirements.

\item

We derive the closed-form expression of the fronthaul load for  our proposed scheme with the special case and establish the upper and lower bounds of the fronthaul load for our proposed scheme with the general case.
  We show that the fronthaul load using our proposed scheme is at most a constant factor larger than that of the  Maddah-Ali-Niesen's decentralized scheme.

\item

We validate our theoretical results by using computer simulations, which show
that our proposed scheme can create considerable coded-multicasting opportunities  in asynchronous request scenarios, and   the maximum request delay can be adjusted flexibly to achieve the load-delay tradeoff.

\end{enumerate}

The rest of this paper is organized as follows. In Section II,  the system model is introduced. Our proposed asynchronous coded caching scheme is presented in Section III. The performance analysis of our proposed scheme is given in Section IV.
In Section V, simulations results are shown. Final conclusions are drawn in Section VI.

\section{System Model}

Consider the F-RAN as shown in Fig. \ref{fran} where there are  $K$ F-APs and each F-AP serves multiple users.\footnote{{The network setting here is similar to that in \cite{Maddah-Ali1}.
However, there exist some difference between them.
In our network setting, users are connected through F-APs to the cloud server, where an F-AP can store large amount of local data and serve multiple users.
Actually, the F-AP can be seen as a sub-server and a request relay station, which can increase the satisfaction of users, improve the efficiency of requests, and reduce the overload of the cloud server.}}
Assume that the users request contents asynchronously during the time interval $\left( {0,T} \right]$.
Let ${\cal K} = \left\{ {1,2,\ldots,k,\ldots,K} \right\}$  denote the index set of the considered $K$ F-APs.
The cloud server has access to a content library of $N$  files, denoted by ${W_1},{W_2},\ldots,{W_N}$.
Let ${\cal N} = \left\{ {1,2,\ldots,n,\ldots,N} \right\}$  denote the index set of the $N$ files with $N \ge K$.
Assume that the size of each file  is $F$  bits and the files in the content library have a uniform popularity distribution.
For each F-AP, only one of its served users requests one file
during the time interval $\left( {0,T} \right]$, while the F-AP informs the cloud server of the request immediately.
For description convenience, we say that $K$ F-APs request contents asynchronously during the considered time interval $\left( {0,T} \right]$, where each F-AP only requests one file.\footnote{{
Multiple requests can also be handled by using our proposed scheme.
Assume each F-AP can request multiple different files during the considered time interval $(0,T]$.
Let $I_{\text{max}}$ denote the maximum allowed number of requested files by one F-AP.
Let $d_{k,i}$ denote the index of the $i$th requested file  by F-AP $k$ during $(0,T]$ for $k\in \cal K$, and $ d_{k,i}= \varnothing $ if F-AP $k$ does not request any file for the $i$th request.
Let $D_1=\left\{d_{1,1},d_{2,1}, \ldots,d_{K,1} \right\}, D_2=\left\{d_{1,2},d_{2,2}, \ldots,d_{K,2} \right\},\ldots,{D_{I_\text{max}}=\left\{d_{1,I_\text{max}},d_{2,I_\text{max}}, \ldots,d_{K,I_\text{max}} \right\}}$.
By using our proposed scheme for $D_1, D_2,\ldots, D_{I_\text{max}}$ parallelly and separately, the multiple-request case can then be handled.}}
 Each F-AP has an isolated (normalized by $F$) cache size $M$ with $0 < M < N$.

\begin{figure}[!t]
\centering 
\includegraphics[width=0.45\textwidth]{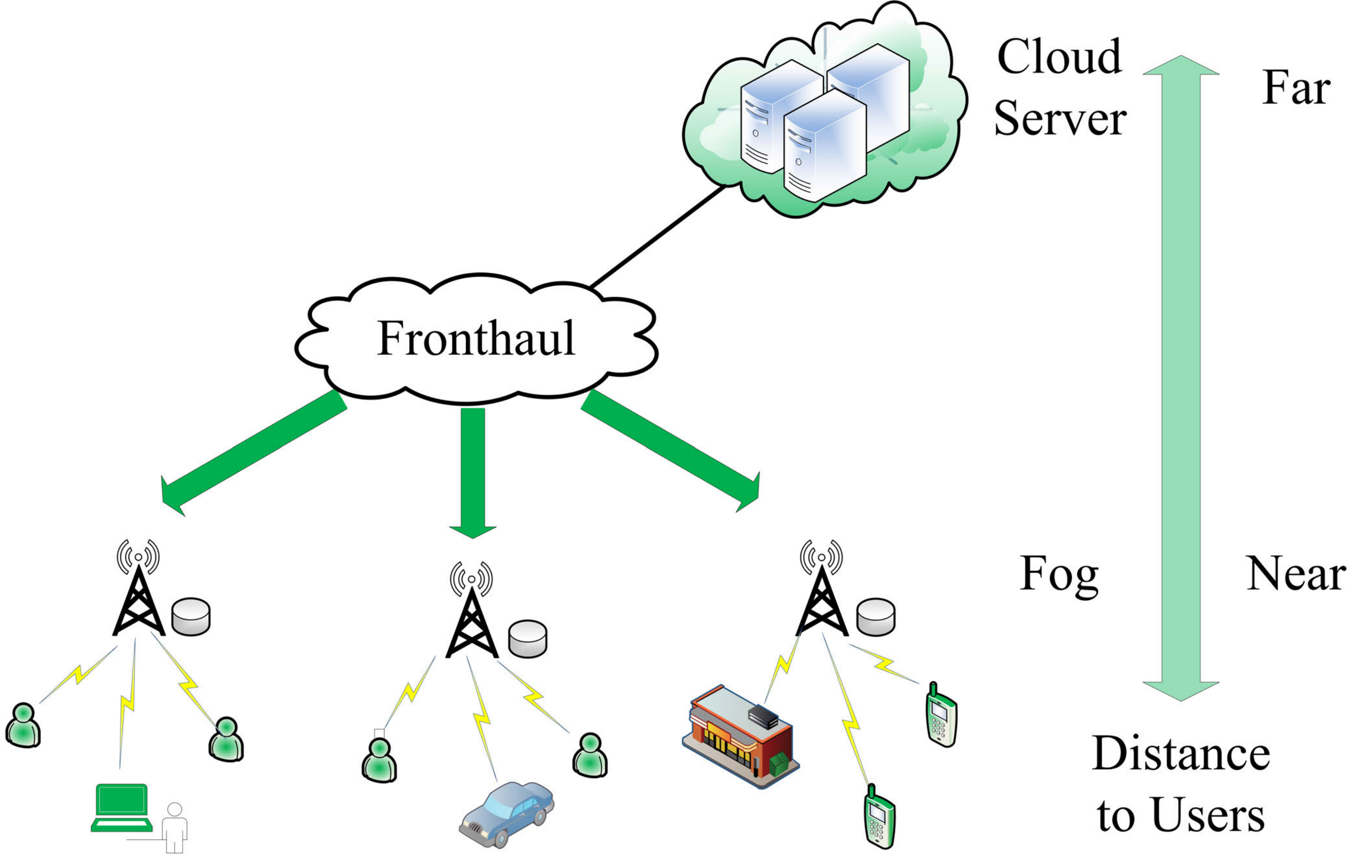}
\caption{Illustration of the asynchronous coded caching scenario in the F-RAN.} \label{fran}
\end{figure}

In the placement phase, the F-APs are given access to the content library.
 By using the same setting as in \cite{Maddah-Ali2}, F-AP $k$ is able to store its cache content $Z_k$ from the content library independently from the other F-APs, i.e., in a decentralized manner.  Let $\phi {}_k$ denote the caching function  of F-AP $k $, which maps the content library into the corresponding cache content as follows
 \begin{equation}
 {Z_k} = \phi {}_k\left( {{W_1},{W_2},\ldots,{W_N}} \right).
 \end{equation}
It can be readily seen that the size of ${Z_k}$ is  $MF$ bits.

In the delivery phase,  the cache contents of all the  F-APs are first communicated to  the cloud server, which are  then noted as  cache records at the cloud server.
Without loss of generality, assume that
the time interval  $\left( {0,T} \right]$ is divided into $B$ time slots with $B \ge 2$.
Let $\Delta t = {T \mathord{\left/
 {\vphantom {T B}} \right.
 \kern-\nulldelimiterspace} B}$ denote the time duration of each time slot.
Then, time
slot $b \in \left\{ {1,2,\ldots,B} \right\}$  represents the time interval $\left( {\left( {b - 1} \right)\Delta t,b\Delta t} \right]$.
Let ${{\cal U}_b} \subseteq {\cal K}$  denote the index set of the F-APs whose requests arrive during time slot  $b$ with ${{\cal U}_b} \ne \varnothing $.
Assume that the cloud server is informed of the requests of the F-APs in ${\cal U}_b$ during time slot $b$, which are processed in a unified manner, i.e., the cloud server transmits the coded-multicasting
content to all the $K$ F-APs through the fronthaul links at the end of each time slot for the online case.
Suppose that the maximum request delay that it takes for an F-AP to recover its requested file is $\Delta b \in \left\{ {1,2,\ldots,B} \right\}$ time slots.
In this paper, we do not consider the time that it takes for the cloud server to transmit the corresponding contents to the F-APs and the time that it takes for each F-AP to transmit the recovered file to the served user. Then, the cloud server should fulfill the requests of the F-APs in ${\cal U}_b$ by the end of the time slot $b + \Delta b - 1$.

Let  ${d_k} \in {\cal N}$ denote the index of the file requested by  F-AP $k$ during $\left( {0,T} \right]$, and ${{\boldsymbol{d}}_b} \in {{\cal N}^{\left| {{{\cal U}_b}} \right|}}$ denote the request vector of the corresponding  F-APs in ${\cal U}_b$.
Let $\psi_b$ denote the encoding function of the cloud server at the end of time slot $b$, which maps the files ${W_1},{W_2},\ldots,{W_N}$, the cache contents ${Z_1},{Z_2},\ldots,{Z_K}$, and the requests ${{\boldsymbol{d}}_b}$
 to the coded-multicasting content as follows
 \begin{equation}
 {X_b} \buildrel \Delta \over = \psi_b \left( {{W_1}, {W_2}, \ldots,{W_N},{Z_1},{Z_2},\ldots,{Z_K},{{\boldsymbol{d}}_b}} \right).
 \end{equation}
 Let $\theta _{k} $ denote the decoding function of F-AP $k$, which maps the received  coded-multicasting contents {${X_b},{X_{b+1}}, \ldots,{X_{b+\Delta b -1}}$,} the cache content $Z_k$, and the request ${{d_k}}$ to the estimate of the requested file ${{W_{{d_k}}}}$ of F-AP $k$ as follows
 {
 \begin{equation}
 {\hat W_{d_k}} = {\theta _{k}}\left( {{X_b},{X_{b+1}}, \ldots, {X_{b+\Delta b -1}}},{Z_k},{d_k} \right).
 \end{equation}
 }
Each F-AP should be able to recover its requested file successfully from its cached content and the received coded-multicasting contents, and then transmit  it to the served user. For every large enough file size $F$, an asynchronous coded caching scheme is  feasible if and only if the worst-case propability of error over all the possible requests ${{\boldsymbol{d}}_1},{{\boldsymbol{d}}_2},\ldots,{{\boldsymbol{d}}_B}$
  satisfies the following condition
\begin{equation}
\mathop {\max }\limits_{{{\boldsymbol{d}}_1},{{\boldsymbol{d}}_2}, \ldots ,{{\boldsymbol{d}}_B}} \mathop {\max }\limits_{k \in \mathcal K} P\left( {{{\hat W}_{{d_k}}} \ne {W_{{d_k}}}} \right) < \varepsilon, \quad \varepsilon > 0.
\end{equation}

The  objective  of this paper is to  find a feasible asynchronous coded caching scheme to minimize the worst-case   normalized  fronthaul load (over all the possible requests ${{\boldsymbol{d}}_1},{{\boldsymbol{d}}_2},\ldots,{{\boldsymbol{d}}_B}$) in the delivery phase for a given maximum request delay \cite{add9}.

\section{Proposed
Decentralized Asynchronous Coded Caching Scheme}

In this section, we first introduce the encoding set collapsing rule.
Then, we show the encoding set partition method. Finally,
we present the proposed decentralized asynchronous coded caching scheme.

\subsection{The Proposed Encoding Set Collapsing Rule}

Asynchronous coded caching and synchronous coded caching are thought to be  under the same condition when their system parameters $M$, $K$, and $N$ are the same.
As the conventional synchronous coded caching scheme under the same condition, such as the  Maddah-Ali-Niesen's decentralized scheme,  ${\cal S} \subseteq {\cal K}$ for any $s = \left| {\cal S} \right| \in \mathcal K$ is called an encoding set if
a single coded-multicasting content can be useful to the F-APs in $\cal S$ simultaneously.
It can be  readily seen that the subset of $\cal S$ is also an encoding set.
In order to differentiate the same subfile in asynchronous  and synchronous coded caching,
let $W_{k,{\cal S}}^{\rm a}$ and $W_{k,{\cal S}}^{\rm s}$ denote the bits of the file requested by F-AP $k$ cached exclusively at the F-APs in ${\cal S}$ for asynchronous and synchronous coded caching, respectively.

Consider that the requests of the F-APs in ${{\cal K}}\backslash {{\cal U}_{{1}}}$  have not arrived yet during time slot $1$. Assume   ${{\cal U}_1} \cap {\cal S} \ne \varnothing $ and  the requests
in ${{\boldsymbol{d}}_1}$
should be fulfilled at the end of time slot 1.
The cloud server needs to transmit a coded-multicasting content which is useful to the F-APs in ${{\cal U}_1} \cap {\cal S}$ at the end of time slot 1.
Thus, we say that the encoding set $\cal S$ in synchronous coded caching collapses into a subset of $\cal S$, i.e., ${{\cal U}_1} \cap {\cal S}$, for transmitting the corresponding coded-multicasting content in asynchronous coded caching.
Recall that by applying the scheme in \cite{Maddah-Ali2}, the coded-multicasting content that the cloud server
transmits for $\cal S$ in synchronous coded caching is
${ \oplus _{k \in {\cal S}}}W_{k,{\cal S}\backslash \left\{ k \right\}}^{\rm s}$,
 where $ \oplus $  denotes the bitwise XOR operation.
Accordingly, the cloud server transmits
${ \oplus _{k \in \left( {{\cal S} \cap {{\cal U}_{{1}}}} \right)}}{W_{k,\left( {{\cal S} \cap {{\cal U}_{{1}}}} \right)\backslash \left\{ k \right\}}^{\rm a}}$ at the end of time slot 1.
According to the above discussions, it is obvious that there exists some relationship, called encoding set collapsing rule, between the coded-multicasting contents in asynchronous and synchronous coded caching.

{During each time slot, let ${\cal U}^\textrm y$ denote the index set of the F-APs  from which the requests have arrived and ${\cal U}^\textrm n = {\cal K} \backslash {{\cal U}^\textrm y} $.}
For any { ${{\cal S}^1} \subseteq {{\cal U}^\textrm y}$  and ${{\cal S}^2} \subseteq {{\cal U}^\textrm n}$ }, the encoding set { ${\cal S} = {{\cal S}^1} \cup {{\cal S}^2}$} in synchronous coded caching collapses into {${\cal S}^1$ }in asynchronous coded caching.
Accordingly,
  { ${ \oplus _{k \in \left( {{{\cal S}^1} \cup {{\cal S}^2}} \right)}}W_{k,\left( {{{\cal S}^1} \cup {{\cal S}^2}} \right)\backslash \left\{ k \right\}}^{\rm s} $} collapses into { ${ \oplus _{k \in {{\cal S}^1}}}{W_{k,\left( {{{\cal S}^1} \cup {{\cal S}^2}} \right)\backslash \left\{ k \right\}}^{\rm a}}$ }, which will be practically transmitted by the cloud server  at the end of {the current time slot} in asynchronous coded caching.

\emph{\textbf{Example 1:}} Assume that $N = 4$, $K = 4$, $M = 2$, $B = 4$,
$T = 4 \ \rm s$, $\Delta t = 1 \ \rm s$, $\Delta b = 2$, ${{\cal U}_b} = \left\{ b \right\}$ and $d_k = k$.
Consider the encoding set $\left\{ {1,2,3,4} \right\}$, and
the corresponding coded-multicasting content in synchronous coded caching is $W_{1,\left\{ {2,3,4} \right\}}^{\rm s} \oplus W_{2,\left\{ {1,3,4} \right\}}^{\rm s} \oplus W_{3,\left\{ {1,2,4} \right\}}^{\rm s} \oplus W_{4,\left\{ {1,2,3} \right\}}^{\rm s}$. Note that the request of F-AP 1 should be fulfilled by the end of time slot $2$.
 $W_{3,\left\{ {1,2,4} \right\}}^{\rm s}$  and  $W_{4,\left\{ {1,2,3} \right\}}^{\rm s}$ cannot be encoded together at the end of time slot $2$ because the requests of F-AP 3 and F-AP 4 have not arrived yet.
Correspondingly, $W_{1,\left\{ {2,3,4} \right\}}^{\rm s} \oplus W_{2,\left\{ {1,3,4} \right\}}^{\rm s} \oplus W_{3,\left\{ {1,2,4} \right\}}^{\rm s} \oplus W_{4,\left\{ {1,2,3} \right\}}^{\rm s}$ collapses into  $W_{1,\left\{ {2,3,4} \right\}}^{\rm a} \oplus W_{2,\left\{ {1,3,4} \right\}}^{\rm a} $, and it will be  practically transmitted at the end of time slot $2$.

\subsection{The Proposed Encoding Set Partition Method}

Utilizing our proposed encoding set collapsing rule, we now consider what contents are transmitted in asynchronous coded caching.
In order to fulfill the requests of the F-APs
in asynchronous coded caching, $\cal S$ may need to be collapsed into a subset of $\cal S$ many times to transmit the corresponding  content at the end of different time slots.
For a given $\Delta b$,  {the requests of the F-APs in 
${{\cal U}_{b - \Delta b + 1}}$
need to be fulfilled by the end of time slot $b$.}
Let ${{\cal U}^ {\rm a}} = \mathop  \cup \nolimits_{i = \max \left\{ {1,b - \Delta b + 1} \right\}}^b {{\cal U}_i}$ denote the index set of the active F-APs {during} time slot $b$.
Then, only the files requested by the F-APs in ${{\cal U}^ {\rm a}}$ can be encoded with each other, which means that $\cal S$ collapses into ${\cal S} \cap {{{\cal U}^ {\rm a}}}$.
Moreover,  minimizing the fronthaul load  is equivalent to partitioning $\cal S$ into the minimum number of nonoverlapping subsets for transmission in asynchronous coded caching.

\begin{algorithm}[!t]
\footnotesize
\caption{The proposed encoding set partition method}
\label{partitionMethod}
\begin{algorithmic}[1]

\STATE Initialize $i $, $\beta$, $\gamma$.
\WHILE{${\cal S} \ne \varnothing $}
\STATE $i =  i + 1$.
\WHILE{${{\cal U}_{\beta + 1}} \cap {\cal S} = \varnothing $}
\STATE $\beta =  \beta + 1$.
\ENDWHILE

\IF{$\gamma - \beta \ge  \Delta b$}
\STATE ${\cal S}_i  = {\cal S} \cap \left( {\mathop  \cup \nolimits_{b = \beta + 1}^{\beta + \Delta b} {{\cal U}_b}} \right)$,
\STATE $\beta = \beta + \Delta b$.
\ELSE
\STATE ${\cal S}_i  = {\cal S} \cap \left( {\mathop  \cup \nolimits_{b = \beta + 1}^\gamma {{\cal U}_b}} \right)$.
\ENDIF
\STATE ${\cal S} = {\cal S}\backslash {\cal S}_i$.
\ENDWHILE
\end{algorithmic}
\end{algorithm}

Let $\left( {\beta \Delta t,\gamma \Delta t} \right]$ denote the active time interval  of ${\cal S}$ if $\left( {\left( {\mathop  \cup \nolimits_{b = 1}^\beta {{\cal U}_b}} \right) \cup \left( {\mathop  \cup \nolimits_{b = \gamma + 1}^B {{\cal U}_b}} \right)} \right) \cap {\cal S} = \varnothing $ and ${{\cal U}_b} \cap {\cal S} \ne \varnothing$ for
$b = \beta  + 1$ and $ b = \gamma $, where
 $\beta$ and $\gamma$ are integers with $0 \le \beta  < \gamma \le {B}$.
Suppose that ${\cal S}$ is partitioned into ${\eta _{\cal S}}\left( {\Delta b} \right)$ subsets for a given $\Delta b$, where ${\eta _{\cal S}}\left( {\Delta b} \right)$ is a function of $\Delta b$.
Let ${\cal S}_i$ denote the $i$-th partitioned encoding subset.
The detailed encoding set partition method is presented in Algorithm \ref{partitionMethod}.

 \emph{\textbf{Example 2:}}
Consider the same setting as \emph{Example 1}.
 Focus on  ${\cal S} = \left\{ {1,3,4} \right\}$ with its active time interval $\left( {0,4} \right]$. First, partition $\cal S$ from time slot $1$, and assign
  ${\cal S} \cap \left( {\mathop  \cup \nolimits_{b = 1}^2 {{\cal U}_b}} \right) = \left\{ 1 \right\}$ to the first encoding subset ${\cal S}_1 $.
  Then, let ${\cal S} = {\cal S}\backslash {\cal S}_1 = \left\{ {3,4} \right\}$.
  Partition $\cal S$ from the earliest time slot where there is  at least one F-AP requesting contents in $\cal S$, i.e., time slot $3$, and assign
${\cal S} \cap \left( {\mathop  \cup \nolimits_{b = 3}^4 {{\cal U}_b}} \right) = \left\{ {3,4} \right\}$ to the second encoding subset ${\cal S}_2 $. Correspondingly, $\left\{ {1,3,4} \right\}$ is partitioned into ${\eta _{\cal S}}\left( {2} \right) = 2$ encoding subsets, i.e., $\left\{ 1 \right\}$ and $\left\{ {3,4} \right\}$.

\emph{Remark 1:} According to Algorithm \ref{partitionMethod}, if $\Delta {b_1} > \Delta {b_2}$, we have
\begin{equation}\label{partition}
{\eta _{\cal S}}\left( {\Delta {b_1}} \right) \le {\eta _{\cal S}}\left( {\Delta {b_2}} \right).
\end{equation}

\emph{Remark 2:} It is possible that there exist some time slots in  $\left( {\beta \Delta t,\gamma \Delta t} \right]$,  during which no F-APs request contents. Let  ${B_{\cal S}} \le B$  denote the  number of the time slots in the active time  interval $\left( {\beta \Delta t,\gamma \Delta t} \right]$. Then,  we have
 \begin{equation}
1 \le {\eta _{\cal S}}\left( {\Delta {b}} \right) \le \left\lceil {{{ B}_{\cal S}}/\Delta b} \right\rceil  \le \left\lceil {B/\Delta b} \right\rceil, \label{partM}
\end{equation}
where $\left\lceil \cdot \right\rceil$ denotes the ceil operation.
Specifically, ${\eta _{\cal S}}\left( {\Delta {b}} \right) = 1$ means that $\cal S$ will not be partitioned and the cloud server only needs to transmit a single coded-multicasting content, that is useful to all the F-APs in $\cal S$, without increasing extra fronthaul load in asynchronous
coded caching.
Moreover, it can be readily seen that  ${\eta _{\cal S}}\left( {\Delta {b}} \right) = 1$ for any $\cal S$ only  when $\Delta b = B$.

\subsection{The Proposed Asynchronous Coded Caching Scheme}

 According to the above discussions,
  we propose the following decentralized asynchronous coded caching scheme which exploits the encoding set collapsing rule and implements the  encoding set partition method.
 In the placement phase, each F-AP randomly selects  ${M}F/{N}$ bits of each file with uniform probability and fetch them to fill its cache, which is the same as  the Maddah-Ali-Niesen's decentralized synchronous coded caching scheme.
Note that the placement procedure does not require any coordination and can be operated in a decentralized manner. More specifically, our proposed scheme can operate in the placement phase with an unknown number of F-APs.
In the delivery phase, we propose the following asynchronous and synchronous transmission methods for the online case, which can be chosen by the cloud server.
Note that asynchronous or synchronous here means that the cloud server transmits the coded-multicasting contents asynchronously or synchronously.

\subsubsection{Asynchronous Transmission Method}

 When $\Delta b < B$,
 the asynchronous transmission method is chosen.
The requests  in {${{\boldsymbol{d}}_{b - \Delta b + 1}}$} need to be fulfilled by the end of time slot {$b$}
in order that sufficient  coded-multicasting opportunities can be created.
If $\Delta b > 1$, no contents {need to be} transmitted
at the end of  time slot $1, 2,\ldots ,\Delta b - 1$, {and}  only the  {corresponding} requests of {the active F-APs} 
need to be fulfilled by the cloud server at the end of the time slots between $\Delta b $ and $ B$.
 For description convenience, we say the subfile ${W_{k,{\cal S}}^{\rm a}}$ is of type $s$ {with $s=|S|$}.
  Thus, the cloud server transmits a single coded-multicasting content for the F-APs in $\cal S$ by encoding the subfiles of type $s-1$\cite{Maddah-Ali2}.
  Similarly, we also say the encoding set $\cal S$ is of type $s$.
  At the end of the current time slot, the cloud server firstly partitions each file $W_n$ into  nonoverlapping subfiles\cite{Maddah-Ali2},  whose sizes are calculated according to the updated cache records.

{During time slot $b - \Delta b + 1$, according to the proposed encoding set collapsing rule in Section III-A, we have: ${{\cal U}^ {\rm y}} = {{\cal U}_{b - \Delta b + 1}}$ and ${{\cal U}^ {\rm n}} = {{\cal K}}\backslash {{\cal U}_{b - \Delta b + 1}}$. For any ${{\cal S}^1} \subseteq {{\cal U}_{b - \Delta b + 1}}$ and ${{\cal S}^2} \subseteq {\cal K} \backslash {{\cal U}_{b - \Delta b + 1}}$, the encoding set ${\cal S} = {{\cal S}^1} \cup {{\cal S}^2} \subseteq \cal K$ for any $ s  \in \mathcal K $ in synchronous coded caching collapses into ${\cal S}^1$ in asynchronous coded caching.
Furthermore, according to the proposed encoding set partition method in Section III-B, only the files requested by the active F-APs in ${\cal U}^{\rm a}$ can be encoded with each other, which means that $\cal S$ collapses into $\cal S \cap {\cal U}^{\rm a}$, i.e., $\left( {{{\cal S}^1} \cup {{\cal S}^2}} \right)\cap {\cal U}^{\rm a}$, by the end of time slot $b$.
Let $\chi  = |{\cal S}^{1} | $ and $s- \chi = | {\cal S}^{2} | $.
Suppose $\chi  \in \left\{ {\underline \chi  ,\underline \chi   + 1, \ldots ,\overline \chi  } \right\}$, where $\underline \chi$ and $\overline \chi$ denote the minimum and maximum of $\chi$,\footnote{{
In order that the coded-multicasting content transmitted to the F-APs in  $\cal S$ by the cloud server is useful to at least one F-AP in  ${{\cal U}_{b - \Delta b + 1}}$, ${{\cal S}^1} \cap {\cal S} \ne \varnothing $, i.e., $\chi \ge 1$, needs to be guaranteed. 
Since ${{\cal S}^2} \subseteq {\cal K} \backslash {{\cal U}_{b - \Delta b + 1}}$, we have $| {\cal S}^{2} |
= s - \chi \le  |{\cal K}|   - \left| {\cal U}_{b - \Delta b + 1}\right|$.
Then, we have $\chi \ge  s + \left| {\cal U}_{b - \Delta b + 1} \right| -  | \cal K |$.
Therefore, $\underline \chi   = \max \left\{ {1,s + \left| {{{\cal U}_{b - \Delta b + 1}}} \right| -  |\cal K|} \right\}$.
Besides, it can be easily verified that $\overline \chi   = \min \left\{ {s,\left| {{{\cal U}_{b - \Delta b + 1}}} \right|} \right\}$.
}} respectively.
By considering any $\chi  \in \left\{ {\underline \chi  ,\underline \chi   + 1, \ldots ,\overline \chi  } \right\}$ with any $s \in \cal K$,
the requests of the F-APs in ${\cal U}_{b - \Delta b + 1}$ can be fulfilled by the end of time slot $b$ with sufficient coded multicasting opportunities being created for the F-APs in $\left( {{{\cal S}^1} \cup {{\cal S}^2}} \right)\cap {\cal U}^{\rm a}$.
Recall that the  F-APs in $\left( {{{\cal S}^1} \cup {{\cal S}^2}} \right)\backslash \left\{ k \right\}$ share a subfile which is not available in the cache content ${{Z_k}}$ and requested by F-AP  $k \in \left( {{{\cal S}^1} \cup {{\cal S}^2}} \right) $.
For any ${\cal S}^1$ and  ${\cal S}^2$ with any $s$, in order to avoid transmitting subfiles repeatedly, no contents need to be transmitted
if ${W_{k,\left( {{{\cal S}^1} \cup {{\cal S}^2}} \right)\backslash \left\{ k \right\}}^{\rm a}} = \varnothing $ for $k \in \left( {\left( {{{\cal S}^1} \cup {{\cal S}^2}} \right) \cap {{{\cal U}^ {\rm a}}}} \right)$.
Otherwise, the cloud server transmits the coded-multicasting content by the end of time slot $b$ as follows
 \[{ \oplus _{k \in \left( {\left( {{{\cal S}^1} \cup {{\cal S}^2}} \right) \cap {{{\cal U}^ {\rm a}}}} \right)}}{W_{k,\left( {{{\cal S}^1} \cup {{\cal S}^2}} \right)\backslash \left\{ k \right\}}^{\rm a}}.\]}

 After the transmission is completed,
  each F-AP in ${{\cal U}_{b - \Delta b + 1}}$ recovers  the desirable subfiles of its requested file.
 Then, each F-AP in ${{\cal U}_{b - \Delta b + 1}}$ transmits  the recovered subfiles and the  corresponding subfiles available in its cache to its served user immediately. Correspondingly, the user can recover the desirable file.
Each F-AP in ${{{\cal U}^ {\rm a}}}\backslash {{\cal U}_{b - \Delta b + 1}}$ also  recovers the corresponding desirable subfiles and  transmits them to its served user at this time.
In addition, the cloud server needs to update the cache records of the active F-APs
  by adding a record of the subfiles recovered by each F-AP in ${{\cal U}^ {\rm a}}$  as its cache content at the end of this time slot.
  Note that updating the cache records has no influence on the cache contents of the F-APs, which stay unchanged in the delivery phase. The cache records can help the cloud server identify  whether the subfile to be transmitted is $\varnothing$  or not in real time before transmission.

\begin{algorithm}[!t]
\footnotesize
\setcounter{algorithm}{1}
\caption{The proposed asynchronous coded caching scheme}
\label{alg}
\begin{algorithmic}[1]
\STATE PLACEMENT
\FOR{ $k \in {\cal K}, n \in {\cal N}$}
\STATE F-AP $k$  independently caches ${MF}/{N}$ bits of file $W_n$, chosen uniformly at random.
\ENDFOR
\\-------------------------------------------------------------
\STATE DELIVERY
\STATE Initialize ${{\cal U}^ {\rm a}} = \varnothing$, $b  = 1$.
\WHILE{$b \le B $}

\IF{$ \Delta b < B$}

\IF{$b \le \Delta b - 1$}
\STATE ${{\cal U}^ {\rm a}} = {{\cal U}^ {\rm a}} \cup {{\cal U}_b}$,
\STATE At the end of time slot $b$,  no contents are transmitted.

\ELSIF{$\Delta b - 1 < b < {B}$}
\STATE ${\cal U}^ {\rm a} = {\cal U}^ {\rm a} \cup {\cal U}_b$.
\FOR{$s = \left| {\cal K} \right|,\left| {\cal K} \right| - 1,\ldots,1$}
\FOR{$\chi  = \max \left\{ {1,s + \left| {{{\cal U}_{b - \Delta b + 1}}} \right| - \left| {\cal K} \right|} \right\}:\min \left\{ {s,\left| {{{\cal U}_{b - \Delta b + 1}}} \right|} \right\}$}
    \FORALL{${\cal S}^1 \subseteq {{\cal U}_{b - \Delta b + 1}},{\cal S}^2 \subseteq {{\cal K}} \backslash {{\cal U}_{b - \Delta b + 1} }:\left| {{\cal S}^1} \right| = \chi ,\left| {{\cal S}^2} \right| = s - \chi$}
        \STATE At the end of time slot $b$,  no contents are transmitted if ${W_{k,\left( {{{\cal S}^1} \cup {{\cal S}^2}} \right)\backslash \left\{ k \right\}}^{\rm a}} = \varnothing $ for  $k \in \left( {\left( {{{\cal S}^1} \cup {{\cal S}^2}} \right) \cap {{\cal U}^ {\rm a}}} \right)$;
        Otherwise, the cloud server sends
        ${ \oplus _{k \in \left( {\left( {{{\cal S}^1} \cup {{\cal S}^2}} \right) \cap {{\cal U}^ {\rm a}}} \right)}}{W_{k,\left( {{{\cal S}^1} \cup {{\cal S}^2}} \right)\backslash \left\{ k \right\}}^{\rm a}}$.
    \ENDFOR
\ENDFOR
\ENDFOR
    \STATE ${{\cal U}^ {\rm a}} = {{\cal U}^ {\rm a}} \backslash {{\cal U}_{b - \Delta b+ 1}}$;

\ELSE
\STATE ${\cal U}^ {\rm a} = {\cal U}^ {\rm a} \cup {\cal U}_b$.
\FOR{$s = \left| {\cal K} \right|,\left| {\cal K} \right| - 1,\ldots,1$}
\FOR{$\chi  = \max \left\{ {1,s + \left| {{{\cal U}^ {\rm a}}} \right| - \left| {\cal K} \right|} \right\}:\min \left\{ {s,\left| {{{\cal U}^ {\rm a}}} \right|} \right\}$}
    \FORALL{${\cal S}^1 \subseteq {{\cal U}^ {\rm a}},{\cal S}^2 \subseteq {{\cal K}} \backslash {{\cal U}^ {\rm a} }:\left| {{\cal S}^1} \right| = \chi ,\left| {{\cal S}^2} \right| = s - \chi$}
        \STATE At the end of time slot $B$,
        the cloud server sends
        ${ \oplus _{k \in {\cal S}^1}}{W_{k,\left( {{{\cal S}^1} \cup {{\cal S}^2}} \right)\backslash \left\{ k \right\}}^{\rm a}}.$

    \ENDFOR
\ENDFOR
\ENDFOR
\ENDIF

\ELSE
\IF{$b \le B - 1$}
\STATE At the end of time slot $b$,  no contents are transmitted.
\ELSE
\FOR{$s = \left| {\cal K} \right|,\left| {\cal K} \right| - 1,\ldots,1$}
    \FORALL{${\cal S} \subseteq {{\cal K}}:\left| {\cal S} \right| = s$}
        \STATE At the end of time slot $B$, the cloud server sends ${ \oplus _{k \in {\cal S}}}{W_{k,{\cal S}\backslash \left\{ k \right\}}^{\rm a}}$.
    \ENDFOR
\ENDFOR
\ENDIF

\ENDIF

\STATE $b = b + 1$.
\ENDWHILE

\end{algorithmic}
\end{algorithm}

At the end of time slot $B$, all the requests of the F-APs in ${\cal U}^ {\rm a}$ should be fulfilled together. Similarly, define $\underline \chi^\prime   = \max \left\{ {1,s + \left| {{{\cal U}^ {\rm a}}} \right| - \left| {\cal K} \right|} \right\}$ and $\overline \chi^\prime   = \min \left\{ {s,\left| {{{\cal U}^ {\rm a}}} \right|} \right\}$.
Focus on ${\cal S}^1 \subseteq {{\cal U}^ {\rm a}}$ with {$\chi = | {{{\cal S}^1}} | $} and  ${\cal S}^2 \subseteq {{\cal K}} \backslash {{\cal U}^ {\rm a} }$ with {$s - \chi = | {{{\cal S}^2}} | $}.
For any {$s \in \cal K$} and
{any} $\chi \in \left\{ {\underline \chi^\prime  ,\underline \chi^\prime   + 1, \ldots ,\overline \chi^\prime  } \right\}$,
the cloud server transmits the coded-multicasting content {by the end of time slot $B$} as follows
\[{ \oplus _{k \in {\cal S}^1}}{W_{k,\left( {{{\cal S}^1} \cup {{\cal S}^2}} \right)\backslash \left\{ k \right\}}^{\rm a}},\]
where all the subfiles ${W_{k,\left( {{{\cal S}^1} \cup {{\cal S}^2}} \right)\backslash \left\{ k \right\}}^{\rm a}}$ are assumed to be zero-padded to the number of bits of the longest subfile in the bit-wise XOR operation.
After that, each F-AP in ${{\cal U}^ {\rm a}}$ recovers  the subfiles of its requested file, and then transmits the recovered subfiles and the subfiles available in its cache to its served user. Correspondingly, the user can recover the desirable file.

\subsubsection{Synchronous Transmission Method}

When $\Delta b = B$, the synchronous transmission method is chosen.
Firstly, no contents {need to be} transmitted
at the end of  time slot $1, 2,\ldots ,B - 1$.
At the end of time slot $B$, for all $\cal S$ with any $s$, the cloud server transmits the coded-multicasting content as follows
  \[{ \oplus _{k \in {\cal S}}}{W_{k,{\cal S}\backslash \left\{ k \right\}}^{\rm a}}.\]
Then, each F-AP transmits all  the subfiles of its requested file to its served user. Correspondingly, the user can recover the desirable file.

The detailed description of our proposed decentralized asynchronous coded caching scheme is presented in Algorithm \ref{alg}.
Note that the problem setting allows
for a vanishing probability of error as $F \to \infty$.

\emph{\textbf{Example 3:}} Consider the same setting as \emph{Example 1}.
It can be readily seen that this corresponds to the worst-case request.
According to Algorithm \ref{alg}, the coded-multicasting contents transmitted by the cloud server at the end of time slot 2, 3, and 4  are illustrated in Table \ref{res2}, Table \ref{res3}, and Table \ref{res4}, respectively.
Note that  $\varnothing $ indicates that no contents are transmitted in the tables.  In addition, subfile ${W_{3,\left\{ {2,4} \right\}}^{\rm a}}$ is actually $\varnothing$ according to the updated cache records.

Still consider the same setting as \emph{Example 1}. We explain here how  Algorithm \ref{alg} implements our proposed encoding set partition method.
 Focus on
  ${\cal S} = \left\{ {1,2,3,4} \right\}$.
 Firstly, no contents are transmitted at the end of time slot $1$. At the end of time slot $2$,
  ${W_{1,\left\{ {2,3,4} \right\}}^{\rm a}} \oplus {W_{2,\left\{ {1,3,4} \right\}}^{\rm a}}$ is transmitted with ${{\cal U}^ {\rm a}} = \left\{ {1,2} \right\}$.
 At the end of time slot $3$, the cloud server decides not to transmit ${W_{2,\left\{ {1,3,4} \right\}}^{\rm a}} \oplus {W_{3,\left\{ {1,2,4} \right\}}^{\rm a}}$ with ${{\cal U}^ {\rm a}} = \left\{ {2,3} \right\}$, since ${W_{2,\left\{ {1,3,4} \right\}}^{\rm a}}$ is $\varnothing $
according to the updated cache records.
Correspondingly, no contents are transmitted.
 Finally,
 ${W_{3,\left\{ {1,2,4} \right\}}^{\rm a}} \oplus {W_{4,\left\{ {1,2,3} \right\}}^{\rm a}}$ is transmitted with ${{\cal U}^ {\rm a}} = \left\{ {3,4} \right\}$
 at the end of time slot $4$. It can be readily seen that  $W_{1,\left\{ {2,3,4} \right\}}^{\rm s} \oplus W_{2,\left\{ {1,3,4} \right\}}^{\rm s} \oplus W_{3,\left\{ {1,2,4} \right\}}^{\rm s} \oplus W_{4,\left\{ {1,2,3} \right\}}^{\rm s}$ is partitioned into two parts of equal size, i.e., ${W_{1,\left\{ {2,3,4} \right\}}^{\rm a}} \oplus {W_{2,\left\{ {1,3,4} \right\}}^{\rm a}}$ and ${W_{3,\left\{ {1,2,4} \right\}}^{\rm a}} \oplus {W_{4,\left\{ {1,2,3} \right\}}^{\rm a}}$, for transmission in our proposed asynchronous coded caching scheme.

\begin{table}[!t]
\small
\centering
\setlength{\abovecaptionskip}{0pt}%
\setlength{\belowcaptionskip}{10pt}%
\caption{The contents transmitted during time slot $2$}
\label{res2}
\scalebox{1}{
\begin{tabular}{|c|c|c|c|c|c|}
\hline
$s$ & $\chi$ & ${\cal S}^1$ & ${\cal S}^2$ & ${\cal U}^ {\rm a}$ & Coded-multicasting content   \\
\hline
4 & 1 & $\left\{ 1 \right\}$ & $\left\{ {2,3,4} \right\}$ & $\left\{ {1,2} \right\}$ & ${W_{1,\left\{ {2,3,4} \right\}}^{\rm a}} \oplus {W_{2,\left\{ {1,3,4} \right\}}^{\rm a}}$  \\
\hline
3 & 1 & $\left\{ 1 \right\}$ & $\left\{ {2,3} \right\}$ & $\left\{ {1,2} \right\}$ & ${W_{1,\left\{ {2,3} \right\}}^{\rm a}} \oplus {W_{2,\left\{ {1,3} \right\}}^{\rm a}}$  \\
\hline
3 & 1 & $\left\{ 1 \right\}$ & $\left\{ {2,4} \right\}$ & $\left\{ {1,2} \right\}$ & ${W_{1,\left\{ {2,4} \right\}}^{\rm a}} \oplus {W_{2,\left\{ {1,4} \right\}}^{\rm a}}$  \\
\hline
3 & 1 & $\left\{ 1 \right\}$ & $\left\{ {3,4} \right\}$ & $\left\{ {1,2} \right\}$ & ${W_{1,\left\{ {3,4} \right\}}^{\rm a}}$  \\
\hline
 2& 1 & $\left\{ 1 \right\}$ & $\left\{ {2} \right\}$ & $\left\{ {1,2} \right\}$ & ${W_{1,\left\{ 2 \right\}}^{\rm a}} \oplus {W_{2,\left\{ 1 \right\}}^{\rm a}}$  \\
\hline
2& 1 & $\left\{ 1 \right\}$ & $\left\{ {3} \right\}$ & $\left\{ {1,2} \right\}$ & ${W_{1,\left\{ 3 \right\}}^{\rm a}} $  \\
\hline
2& 1 & $\left\{ 1 \right\}$ & $\left\{ {4} \right\}$ & $\left\{ {1,2} \right\}$ & ${W_{1,\left\{ 4 \right\}}^{\rm a}}$  \\
\hline
1& 1 & $\left\{ 1 \right\}$ & $\varnothing$ & $\left\{ {1,2} \right\}$ & ${W_{1,\varnothing }^{\rm a}}$  \\
\hline
\end{tabular}}
\end{table}

\begin{table}[!t]
\small
\centering
\setlength{\abovecaptionskip}{0pt}%
\setlength{\belowcaptionskip}{10pt}%
\caption{The contents transmitted during time slot $3$}
\label{res3}
\scalebox{1}{
\begin{tabular}{|c|c|c|c|c|c|}
\hline
$s$ & $\chi$ & ${\cal S}^1$ & ${\cal S}^2$ & ${\cal U}^ {\rm a}$ & Coded-multicasting content   \\
\hline
4 & 1 & $\left\{ 2 \right\}$ & $\left\{ {1,3,4} \right\}$ & $\left\{ {2,3} \right\}$ & $\varnothing $  \\
\hline
3 & 1 & $\left\{ 2 \right\}$ & $\left\{ {1,3} \right\}$ & $\left\{ {2,3} \right\}$ & $\varnothing $ \\
\hline
3 & 1 & $\left\{ 2 \right\}$ & $\left\{ {1,4} \right\}$ & $\left\{ {2,3} \right\}$ & $\varnothing $  \\
\hline
3 & 1 & $\left\{ 2 \right\}$ & $\left\{ {3,4} \right\}$ & $\left\{ {2,3} \right\}$ & ${W_{2,\left\{ {3,4} \right\}}^{\rm a}} \oplus {W_{3,\left\{ {2,4} \right\}}^{\rm a}}$  \\
\hline
 2& 1 & $\left\{ 2 \right\}$ & $\left\{ {1} \right\}$ & $\left\{ {2,3} \right\}$ & $\varnothing $  \\
\hline
2& 1 & $\left\{ 2 \right\}$ & $\left\{ {3} \right\}$ & $\left\{ {2,3} \right\}$ & ${W_{2,\left\{ 3 \right\}}^{\rm a}} \oplus {W_{3,\left\{ 2 \right\}}^{\rm a}}$  \\
\hline
2& 1 & $\left\{ 2 \right\}$ & $\left\{ {4} \right\}$ & $\left\{ {2,3} \right\}$ & ${W_{2,\left\{ 4 \right\}}^{\rm a}}$  \\
\hline
1& 1 & $\left\{ 2 \right\}$ & $\varnothing$ & $\left\{ {2,3} \right\}$ & ${W_{2,\varnothing }^{\rm a}}$  \\
\hline
\end{tabular}}
\end{table}

\begin{table}[!t]
\small
\centering
\setlength{\abovecaptionskip}{0pt}%
\setlength{\belowcaptionskip}{10pt}%
\caption{The contents transmitted during time slot $4$}
\label{res4}
\scalebox{1}{
\begin{tabular}{|c|c|c|c|c|c|}
\hline
$s$ & $\chi$ & ${\cal S}^1$ & ${\cal S}^2$ & ${\cal U}^ {\rm a}$ & Coded-multicasting content  \\
\hline
4 & 2 & $\left\{ {3,4} \right\}$ & $\left\{ {1,2} \right\}$ & $\left\{ {3,4} \right\}$ & ${W_{3,\left\{ {1,2,4} \right\}}^{\rm a}} \oplus {W_{4,\left\{ {1,2,3} \right\}}^{\rm a}}$  \\
\hline
3 & 1 & $\left\{ 3 \right\}$ & $\left\{ {1,2} \right\}$ & $\left\{ {3,4} \right\}$ & ${W_{3,\left\{ {1,2} \right\}}^{\rm a}}$ \\
\hline
3 & 1 & $\left\{ 4 \right\}$ & $\left\{ {1,2} \right\}$ & $\left\{ {3,4} \right\}$ & ${W_{4,\left\{ {1,2} \right\}}^{\rm a}}$  \\
\hline
3 & 2 & $\left\{ {3,4} \right\}$ & $\left\{ {1} \right\}$ & $\left\{ {3,4} \right\}$ & ${W_{3,\left\{ {1,4} \right\}}^{\rm a}} \oplus {W_{4,\left\{ {1,3} \right\}}^{\rm a}}$  \\
\hline
3 & 2 & $\left\{ {3,4} \right\}$ & $\left\{ {2} \right\}$ & $\left\{ {3,4} \right\}$ & ${W_{3,\left\{ {2,4} \right\}}^{\rm a}}\left( \varnothing  \right) \oplus {W_{4,\left\{ {2,3} \right\}}^{\rm a}}$  \\
\hline
 2& 1 & $\left\{ 3 \right\}$ & $\left\{ {1} \right\}$ & $\left\{ {3,4} \right\}$ & ${W_{3,\left\{ 1 \right\}}^{\rm a}}$  \\
\hline
2& 1 & $\left\{ 3 \right\}$ & $\left\{ {2} \right\}$ & $\left\{ {3,4} \right\}$ & $\varnothing$  \\
\hline
2& 1 & $\left\{ 4\right\}$ & $\left\{ {1} \right\}$ & $\left\{ {3,4} \right\}$ & ${W_{4,\left\{ 1 \right\}}^{\rm a}}$  \\
\hline
2& 1 & $\left\{ 4 \right\}$ & $\left\{ {2} \right\}$ & $\left\{ {3,4} \right\}$ & ${W_{4,\left\{ 2 \right\}}^{\rm a}}$  \\
\hline
2& 2 & $\left\{ {3,4} \right\}$ & $\varnothing$ & $\left\{ {3,4} \right\}$ & ${W_{3,\left\{ 4 \right\}}^{\rm a}} \oplus {W_{4,\left\{ 3 \right\}}^{\rm a}}$  \\
\hline
1& 1 & $\left\{ 3 \right\}$ & $\varnothing$ & $\left\{ {3,4} \right\}$ & ${W_{3,\varnothing }^{\rm a}}$  \\
\hline
1& 1 & $\left\{ 4 \right\}$ & $\varnothing$ & $\left\{ {3,4} \right\}$ & ${W_{4,\varnothing }^{\rm a}}$  \\
\hline
\end{tabular}}
\end{table}

\newcounter{TempEqCnt} 
\setcounter{TempEqCnt}{\value{equation}} 
\setcounter{equation}{8} 
\begin{figure*}[hb]
\hrulefill
\begin{equation}\label{q}
q\left( {s,Y,\Delta b} \right) = \left\{ {\begin{array}{*{20}{l}}
{\left( {\begin{array}{*{20}{c}}
K\\
{s }
\end{array}} \right),}&{\Delta b = B,}\\
{{q_2}\left( {s,Y,\Delta b} \right),}&{\Delta b = 1,}\\
{{q_2}\left( {s,Y,\Delta b} \right),}&{\Delta b = \max \left\{ {\left\lceil {\frac{{s  - \left( {Y - 1} \right)L}}{L}} \right\rceil ,1} \right\}< B,}\\
{\sum\limits_{\Delta {b^\prime } = \max \left\{ {\left\lceil {\frac{{s  - \left( {Y - 1} \right)L}}{L}} \right\rceil ,1} \right\}}^{\Delta b - 1} {{q_1}\left( {s,Y,\Delta {b^\prime },\Delta b} \right)}  + {q_2}\left( {s,Y,\Delta b} \right),}&{\max \left\{ {\left\lceil {\frac{{s - \left( {Y - 1} \right)L}}{L}} \right\rceil ,1} \right\} < \Delta b \le B - \left( {Y - 1} \right)\Delta b,}\\
{\sum\limits_{\Delta {b^\prime } = \max \left\{ {\left\lceil {\frac{{s  - \left( {Y - 1} \right)L}}{L}} \right\rceil ,1} \right\}}^{B - \left( {Y - 1} \right)\Delta b} {{q_1}\left( {s,Y,\Delta {b^\prime },\Delta b} \right)} ,}&{B - \left( {Y - 1} \right)\Delta b < \Delta b < B.}
\end{array}} \right.
\end{equation}
\end{figure*}
\setcounter{equation}{\value{TempEqCnt}} 

\emph{Remark 3:}
The major innovation of our proposed scheme is
to partition the coded-multicasting contents in synchronous coded caching by using our proposed encoding set partition method.
{{Our contributions are mainly reflected in selecting a part of the file library for further processing based on the request arrival status of the considered time slot.}}
Our proposed scheme  can create
{considerable coded-multicasting opportunities}
while the maximum request delay of each F-AP is no more than $\Delta b$ time slots.

\emph{Remark 4:} {
Both Maddah-Ali-Niesen's decentralized 
scheme in \cite{Maddah-Ali2}  and our proposed scheme are time-slot based ones.
However, in Maddah-Ali-Niesen's asynchronous coded caching scheme, only simple extension from synchronous coded caching is considered. Just as shown in the provided example, during the first and last two time slots, no coded multicasting opportunities have been created.
In comparison, in our proposed scheme, we try to create sufficient coded multicasting opportunities under the constraint of the maximum request delay $\Delta b$.
Specifically, when $\Delta b >1$, the cloud server does not transmit any content at the end of time slot $1, 2, \ldots, \Delta b-1$ in order that more coded multicasting opportunities can be created in the subsequent time slots.
Besides, at the end of the last time slot $B$, coded multicasting opportunities have been created to fulfil all the requests of the active F-APs.
}

\emph{Remark 5:} {In \cite{Ghasemi}, the asynchronous coded caching problem has been considered mainly for the offline case, where the authors  presented their proposed approach based on a system using the centralized synchronous coded caching scheme in \cite{Maddah-Ali1}. Furthermore, the authors proposed a linear programming formulation that minimizes the overall rate from the server subject to the constraint that each user meets its deadline.
In comparison, we  propose a decentralized asynchronous coded caching scheme based on a different system model for the online case, which is more applicable for practical scenarios, and can work well for both the online case and offline case.
Moreover, we propose an encoding set collapsing rule and an encoding set partition method to minimize the worst-case normalized fronthual load subject to a given maximum request delay.}

\emph{Remark 6:} In the delivery phase, the cloud server needs to partition $N$ files according to the updated cache records at the end of time slot $\Delta b,\Delta b + 1,\ldots,B $. In practice, the cloud server only needs to partition $N$ files based on the initial cache records before starting transmission, and updates the sizes of the  subfiles in the cache records that are encoded to be transmitted by setting them to $\varnothing$ directly after completing transmission at the end of each time slot.

\section{Performance Analysis of the Proposed Asynchronous Coded Caching Scheme}

To emphasize the dependence of the fronthaul load on the cache size $M$, the number of files $N$,
the number of F-APs $K$ and the maximum request delay, let $R_A\left( {M,N,K,\Delta b} \right)$ denote the fronthaul load of our proposed asynchronous coded caching scheme.
As the  request distribution during the $B$ time slots affects $R_A\left( {M,N,K,\Delta b} \right)$, it is hard to obtain its exact expression. Focusing on the special case with $\left| {{{\cal U}_b}} \right| = L$, we can derive the closed-form expression of $R_A\left( {M,N,K,\Delta b} \right)$. As for the general case with random $\left| {{{\cal U}_b}} \right|$, we can establish the upper and lower bounds of $R_A\left( {M,N,K,\Delta b} \right)$.

\subsection{Special Case with $\left| {{{\cal U}_b}} \right| = L$}

Let $\left| {{W_{k,{{{\cal S} \backslash \left\{ k \right\}}}}^{\rm a}}} \right|$ denote the size of ${{W_{k,{{{\cal S} \backslash \left\{ k \right\}}}}^{\rm a}}}$.
According to the law of large numbers, we have
\begin{equation} \label{subfile}
\left| {W_{k,{\cal S}\backslash \left\{ k \right\}}^{\rm{a}}} \right| \approx {\left( {M/N} \right)^{s - 1}}{\left( {1 - M/N} \right)^{K - \left( {s - 1} \right)}}F.
\end{equation}

\begin{theorem} \label{t1}

Consider the special case with
$\left| {{{\cal U}_b}} \right| = L$ and
$B \ge 3$.
For  large enough $F$, the fronthaul load is arbitrarily close to
\begin{equation}\label{theo1}
\begin{split}
R_A\left( {M,N,K,\Delta b} \right) & = \sum\limits_{s = 1}^{K} {\left| {W_{k,{{\cal S} \backslash \left\{ k \right\}}}^{\rm{a}}} \right|\sum\limits_{Y = \left\lceil {\frac{{s}}{{\Delta b \cdot L}}} \right\rceil }^{\min \left\{ {\left\lceil {\frac{B}{{\Delta b}}} \right\rceil ,s } \right\}} {q\left( {s,Y,\Delta b} \right)} Y},
\end{split}
\end{equation}
where $q\left( {s,Y,\Delta b} \right)$ is shown at the bottom of this page with
\setcounter{equation}{9}
\begin{small}
\begin{equation}\label{q1Theo1}
\begin{split}
&{q_1}\left( {s,Y,{{\Delta b}^\prime} ,\Delta b} \right) \\
& = \begin{array}{*{20}{l}}
{\left( {\begin{array}{*{20}{c}}
{B - {{\Delta b}^\prime}  + \left( {Y - 1} \right)\left( {\Delta b - 1} \right)}\\
{Y - 1}
\end{array}} \right)}\\
{  \times  \sum\limits_{\alpha  = \max \left\{ {Y,s  - \left( {\left( {Y - 1} \right)\Delta b + {{\Delta b}^\prime}  - Y} \right)L} \right\}}^{\min \left\{ {s ,YL} \right\}} {b\left( {Y,\alpha } \right)\left( {\begin{array}{*{20}{c}}
{\left( {\left( {Y - 1} \right)\Delta b + {{\Delta b}^\prime}  - Y} \right)L}\\
{s - \alpha  }
\end{array}} \right)} ,}
\end{array}
 \end{split}
 \end{equation}
 \end{small}
\begin{small}
\begin{equation} \label{q2}
\begin{split}
&{q_2}\left( {s,Y,\Delta b} \right) \\
 &= \left\{ {\begin{array}{*{20}{l}}
{\left( {\begin{array}{*{20}{c}}
B\\
Y
\end{array}} \right)b\left( {Y,s} \right),}&{\Delta b = 1,}\\
{\begin{array}{*{20}{l}}
{\left( {\begin{array}{*{20}{c}}
{B - Y\left( {\Delta b - 1} \right)}\\
Y
\end{array}} \right)}\\
{  \times  \sum\limits_{\alpha  = \max \left\{ {Y,s - Y\left( {\Delta b - 1} \right)L} \right\}}^{\min \left\{ {s,Y  L} \right\}} {b\left( {Y,\alpha } \right)\left( {\begin{array}{*{20}{c}}
{Y\left( {\Delta b - 1} \right)L}\\
{s  - \alpha }
\end{array}} \right)} ,}
\end{array}}&{1 < \Delta b < B,}
\end{array}} \right.
 \end{split}
 \end{equation}
 \end{small}
 \begin{small}
 \begin{equation} \label{b}
b\left( {Y,\alpha } \right) = \left\{ {\begin{array}{*{20}{l}}
{\left( {\begin{array}{*{20}{c}}
L\\
\alpha
\end{array}} \right),}&{Y = 1,1 \le \alpha  \le L,}\\
{{{\left( {\begin{array}{*{20}{c}}
L\\
1
\end{array}} \right)}^Y},}&{Y > 1,\alpha  = Y,}\\
{1,}&{Y > 1,L > 1,\alpha  = YL,}\\
{\sum\limits_{v = 1}^{\min \left\{ {L,\alpha  - \left( {Y - 1} \right)} \right\}} {\left( {\begin{array}{*{20}{c}}
L\\
v
\end{array}} \right)b\left( {Y - 1,\alpha  - v} \right),} }&{\rm{else.}}
\end{array}} \right.
 \end{equation}
 \end{small}

\end{theorem}

\begin{proof}
Please see appendix A.
\end{proof}

\emph{Remark 6:}
when $B = 2$, we have
\begin{equation}
q\left( {s,Y,\Delta b} \right) = \left\{ {\begin{array}{*{20}{c}}
{\left( {\begin{array}{*{20}{c}}
K\\
{s }
\end{array}} \right),}&{\Delta b = 2,}\\
{{q_2}\left( {s,Y,\Delta b} \right),}&{\Delta b = 1.}
\end{array}} \right.
\end{equation}
For description convenience,  the case where $B = 2$ is  omitted in Theorem \ref{t1}.

Let ${R_S}\left( M,N,K \right)$ denote the fronthaul load of the Maddah-Ali-Niesen's decentralized synchronous coded caching scheme. From \cite{Maddah-Ali2}, we have
\begin{equation}\label{syn_load}
\begin{split}
 {R_S}\left( M,N,K \right)
   & = F \sum\limits_{s = 1}^{K} {\left( {\begin{array}{*{20}{c}}
K\\
s
\end{array}} \right){{\left( {M/N} \right)}^{s - 1 }}{{\left( {1 - M/N} \right)}^{K - \left( {s - 1} \right) }}}      \\
   &  = FK(1 - M/N)\frac{N}{{KM{\rm{ }}}}\left( {1 - {{\left( {1 - {M \mathord{\left/
 {\vphantom {M N}} \right.
 \kern-\nulldelimiterspace} N}} \right)}^K}} \right)
 , \quad N \ge K.
 \end{split}
\end{equation}

According to Theorem \ref{t1}, we also have
\begin{equation}
\begin{split}
 {\left. {{R_A}\left( {M,N,K,\Delta b} \right)} \right|_{\Delta b = B}} &= F\sum\limits_{s = 1}^{K } {\left| {W_{k,
{{\cal S} \backslash \left\{ k \right\}}}^{\rm{a}}} \right|\sum\limits_{Y = \left\lceil {\frac{{s }}{BL}} \right\rceil }^{\min \left\{ {1,s } \right\}} {q\left( {s,Y,\Delta b} \right)} Y}  \\
 &= F\sum\limits_{s = 1}^{K} {\left| {W_{k,
{{\cal S} \backslash \left\{ k \right\}}}^{\rm{a}}} \right|\sum\limits_{Y = 1}^1 {q\left( {s,Y,\Delta b} \right)} Y} \\
 & = F\sum\limits_{s = 1}^{K } {{{\left( {M/N} \right)}^{s - 1}}{{\left( {1 - M/N} \right)}^{K - \left( {s - 1} \right)}}\left( {\begin{array}{*{20}{c}}
K\\
{s}
\end{array}} \right)}  \\
&= {R_S}\left( {M,N,K} \right), \quad N \ge K.
 \end{split}
\end{equation}
It can be readily seen that our proposed scheme has the same fronthaul load as  the  Maddah-Ali-Niesen's decentralized scheme when ${\Delta b = B}$.

\subsection{General Case with Random $\left| {{{\cal U}_b}} \right|$}

\begin{theorem} \label{t2}
 The fronthaul load of  our proposed scheme  is bounded as follows
 \begin{small}
 \begin{equation}\label{theo2}
 \begin{split}
{R_S}\left( {M,N,K} \right) & \le {R_A}\left( {M,N,K,\Delta b} \right) \\
& \le FK(1 - \frac{M}{N})\min \left\{ {\left\lceil {\frac{B}{{\Delta b}}} \right\rceil \frac{N}{{KM}}\left( {1 - {{\left( {1 - {M \mathord{\left/
 {\vphantom {M N}} \right.
 \kern-\nulldelimiterspace} N}} \right)}^K}} \right),1} \right\}.
 \end{split}
\end{equation}
\end{small}
\end{theorem}

\begin{proof}
Please see appendix B.
\end{proof}

{
It can be readily seen that the lower bound of the fronthual load of our proposed scheme is the same as that of the Maddah-Ali-Niesen's decentralized synchronous coded caching scheme, and also does not consider the asynchronous case, which is of great challenge.
According to (15), our proposed scheme has the same fronthaul load as the Maddah-Ali-Niesen's decentralized scheme when $\Delta b = B $ and $\left| {{{\cal U}_b}} \right| = L$.
It can be seen that the lower bound can indeed be achieved for the above special case.
However, as can be seen from (8)-(12), even for the special case with $\left| {{{\cal U}_b}} \right| = L$,
the expression of the fronthual load of our proposed scheme, i.e., $R_A (M,N,K,\Delta b)$, is already extremely complicated.
Therefore, we can readily see that it will be very difficult to obtain the expression or the lower bound (considering the asynchronous case) of $R_A (M,N,K,\Delta b)$ for the general case with random $\left| {{{\cal U}_b}} \right|$.
}

Comparing the fronthaul load  of our proposed scheme  with that of the Maddah-Ali-Niesen's decentralized scheme, we have the following corollary.

\begin{corollary} \label{cor}
\begin{equation}
1 \le \frac{{R_A\left( {M,N,K,\Delta b} \right)}}{{{R_S}\left( M,N,K \right)}} \le \left\lceil {B/\Delta b} \right\rceil, \quad N \ge K.
\end{equation}
\end{corollary}

As for the general case of asynchronous requests, Theorem \ref{t2} and Corollary \ref{cor} show that the fronthaul load of our proposed scheme is at most a factor $\left\lceil {B/\Delta b} \right\rceil $ larger than that of the Maddah-Ali-Niesen's decentralized  scheme.
When $\Delta b < B$, the performance gap between our proposed scheme and the Maddah-Ali-Niesen's decentralized  scheme is due to the fact that asynchronous requests  lead to the loss of coded-multicasting opportunities, which is also the key difference between asynchronous and synchronous coded caching.

According to  (\ref{partition}), the number of partitioned  encoding subsets for an encoding set with a smaller $\Delta b$ is relatively larger, which can be  illustrated in Fig. \ref{encoding_set}. Note that $R_A\left( {M,N,K,\Delta b} \right)$ can be calculated by accumulating the sizes of  the  coded-multicasting contents corresponding to the subsets that all the encoding sets are partitioned into.
Correspondingly, $R_A\left( {M,N,K,\Delta b} \right)$  increases
with $\Delta b$, since the size of the coded-multicasting content transmitted for an encoding set by the cloud server is unchanged.
Moreover, for $M \in \left[ {{N \mathord{\left/
 {\vphantom {N K}} \right.
 \kern-\nulldelimiterspace} K},N} \right]$,  the fronthaul load of the Maddah-Ali-Niesen's  decentralized scheme can be up to a factor $K$ smaller than that of the uncoded caching scheme \cite{Maddah-Ali2}. Besides, when the number of F-APs increases for $M \in \left[ {{N \mathord{\left/
 {\vphantom {N K}} \right.
 \kern-\nulldelimiterspace} K},N} \right]$, the maximum request delay can be set to a relatively smaller value with the same considerable coded-multicasting opportunities created.

\begin{figure}[!t]
\centering 
\includegraphics[width=0.5\textwidth]{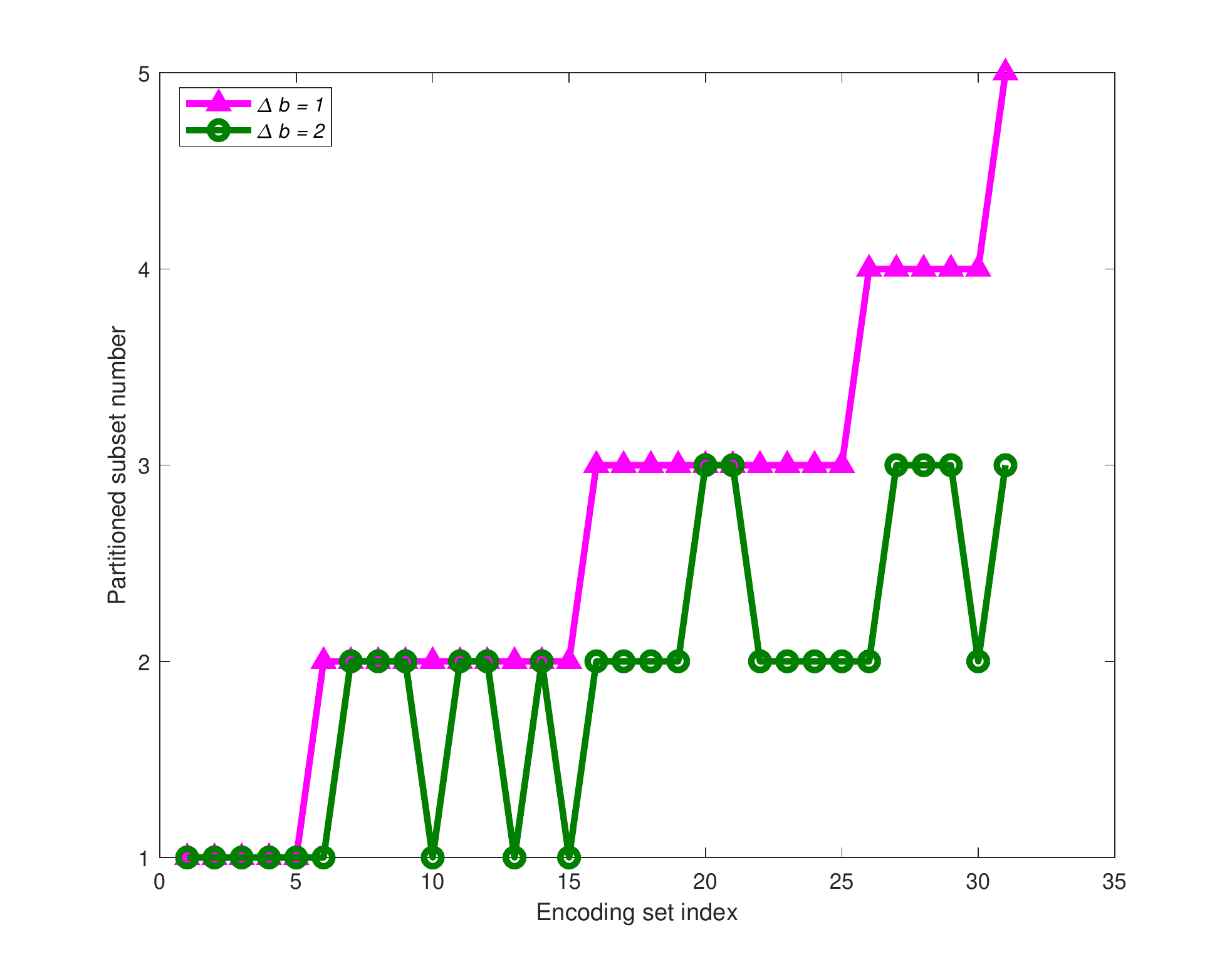}
\caption{The number of  encoding subsets such that each encoding set can be partitioned into with  $B = 5$ and  ${\cal U}_b = \left\{ b \right\}$.} \label{encoding_set}
\end{figure}

\section{Simulation Results}

    In this section, the performance of our  proposed decentralized asynchronous coded caching scheme is evaluated via simulations.
We adopt  the Maddah-Ali-Niesen's decentralized scheme and the uncoded caching scheme as   baselines.
The system parameters are set as follows: $F = 1 \  \rm Gb$, $N = 100$,  $T = 10 \ \rm s$, $B = 5$.

\begin{figure}[!t]
\centering 
\includegraphics[width=0.5\textwidth]{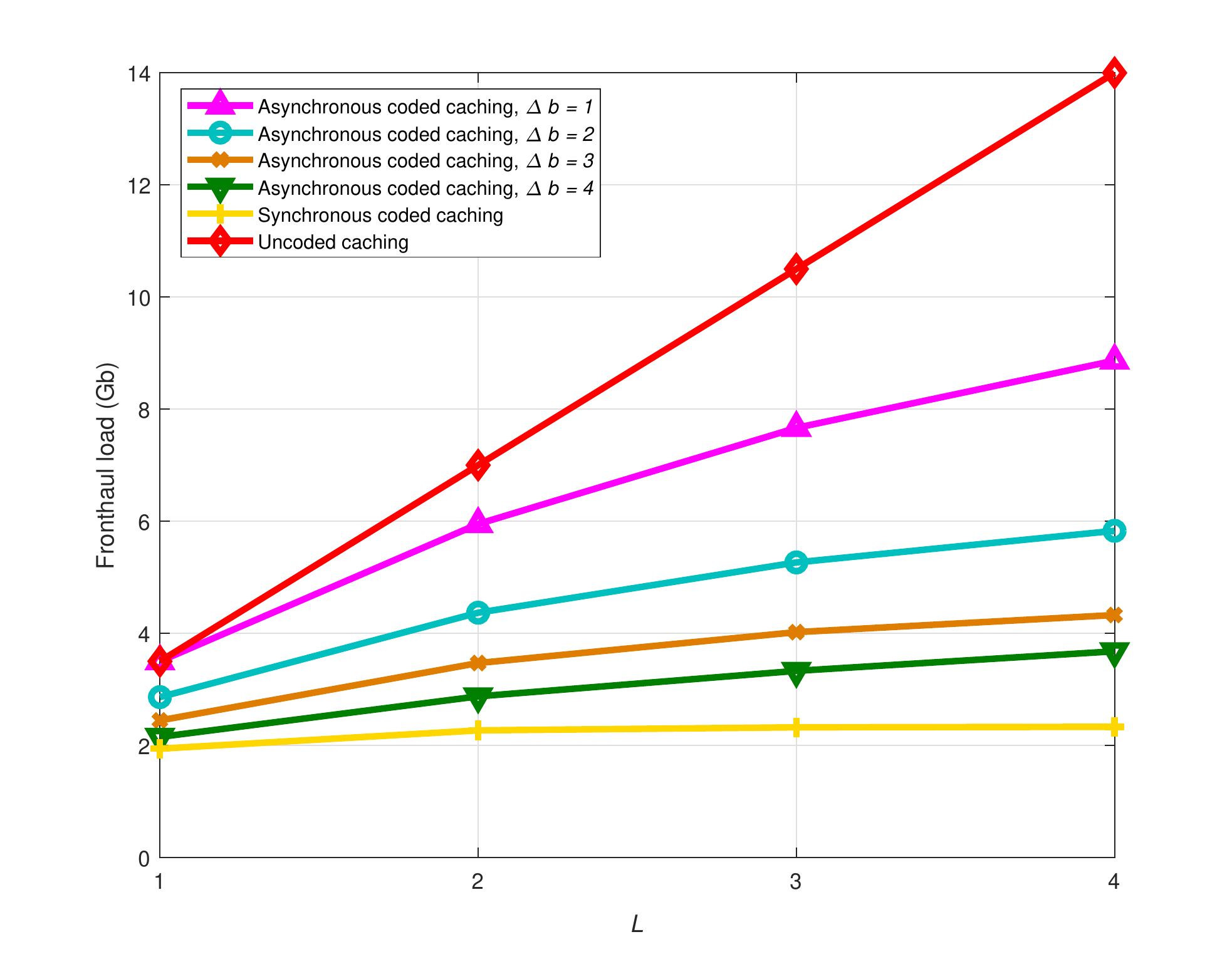}
\caption{ Fronthaul load versus $L$ with  $M=30$. } \label{user_load}
\end{figure}

\begin{small}
\begin{figure}[!t]
\centering 
\includegraphics[width=0.5\textwidth]{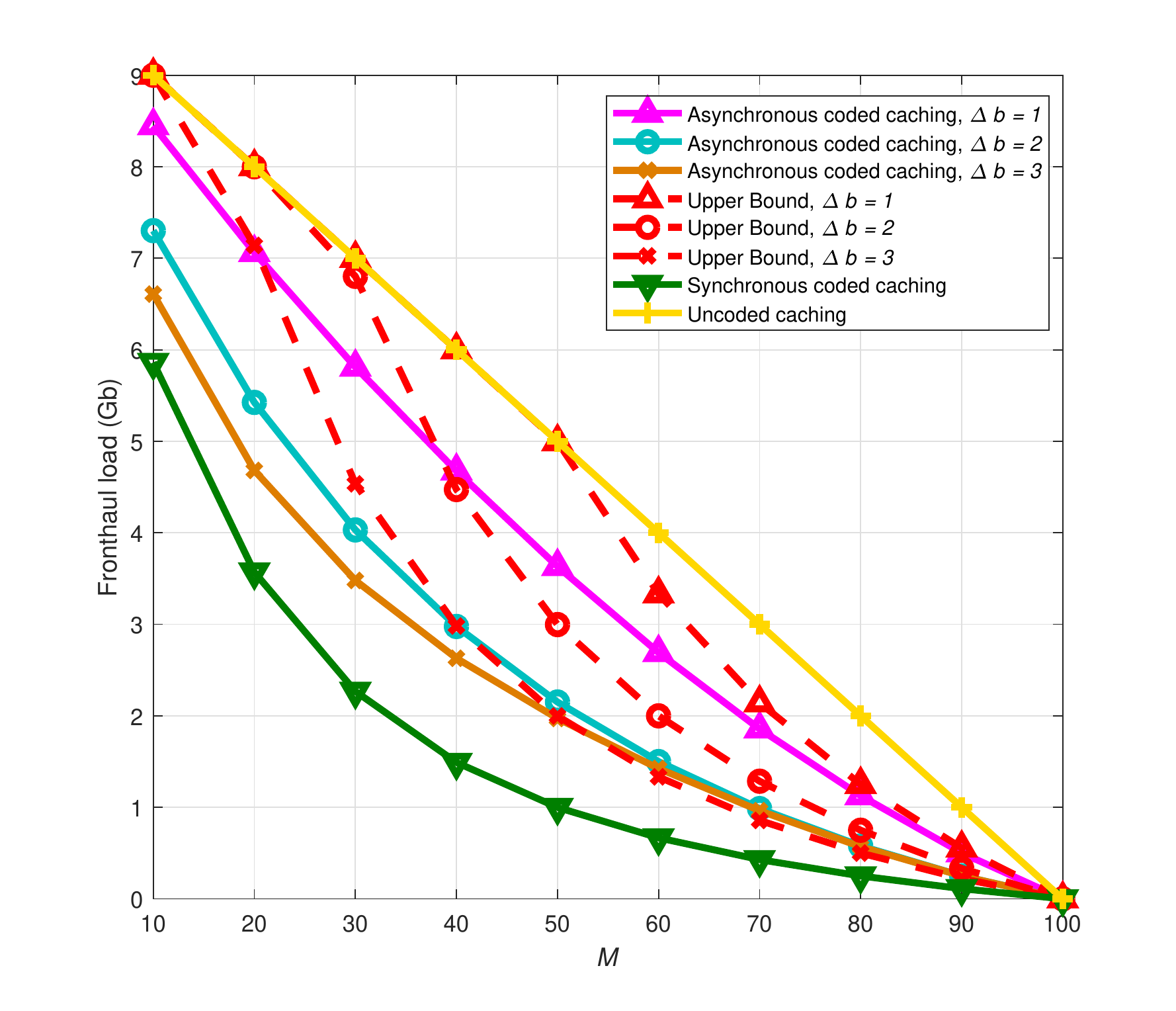}
\caption{ Fronthaul load versus $M$ with $K = 10$.} \label{cache_load}
\end{figure}
\end{small}

\begin{figure}[!t]
\centering 
\includegraphics[width=0.47\textwidth]{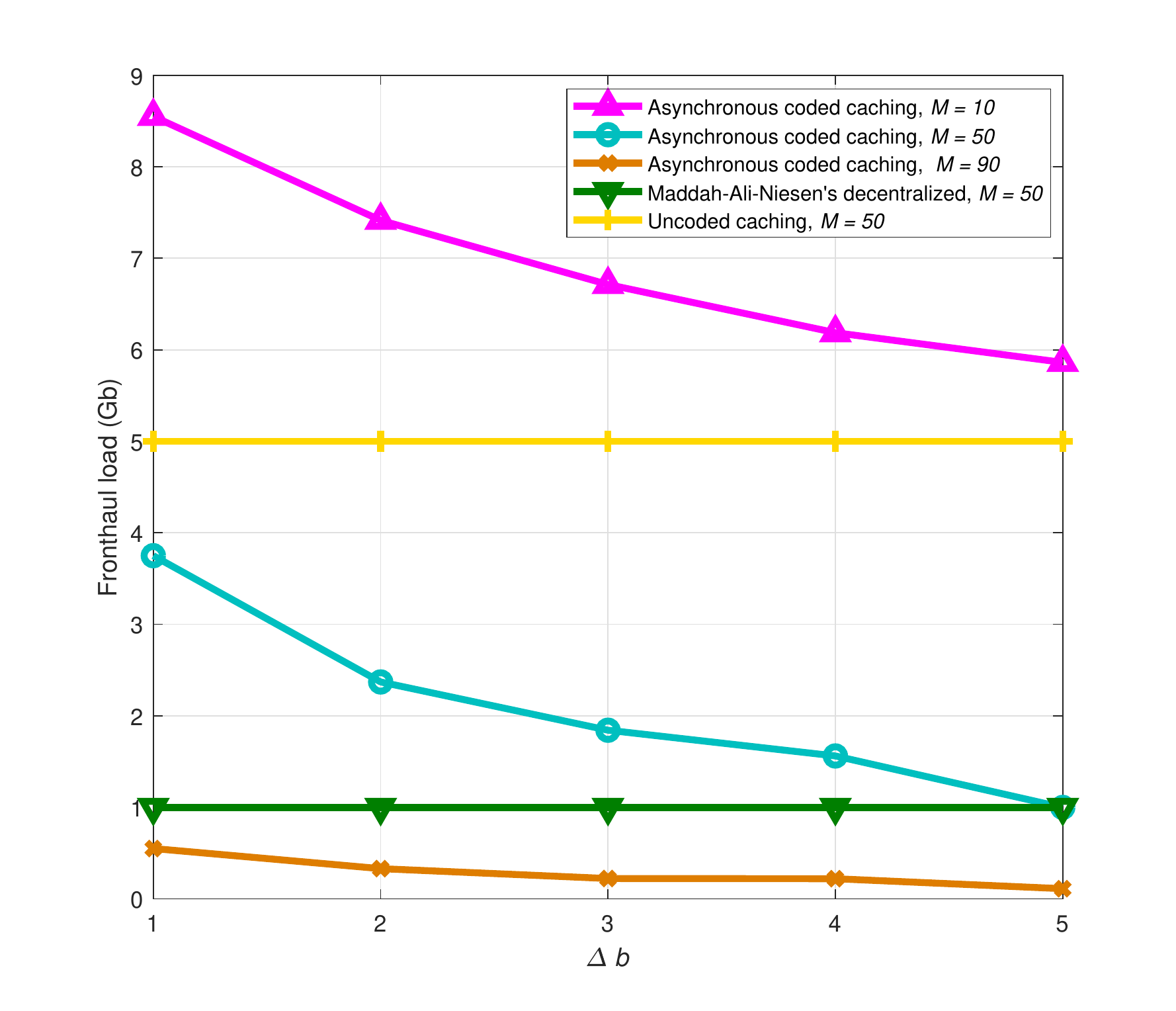}
\caption{ Fronthaul load versus $\Delta b$ for varying cache sizes with $K = 10$. } \label{joint}
\end{figure}

\begin{figure}[!t]
\centering 
\includegraphics[width=0.48\textwidth]{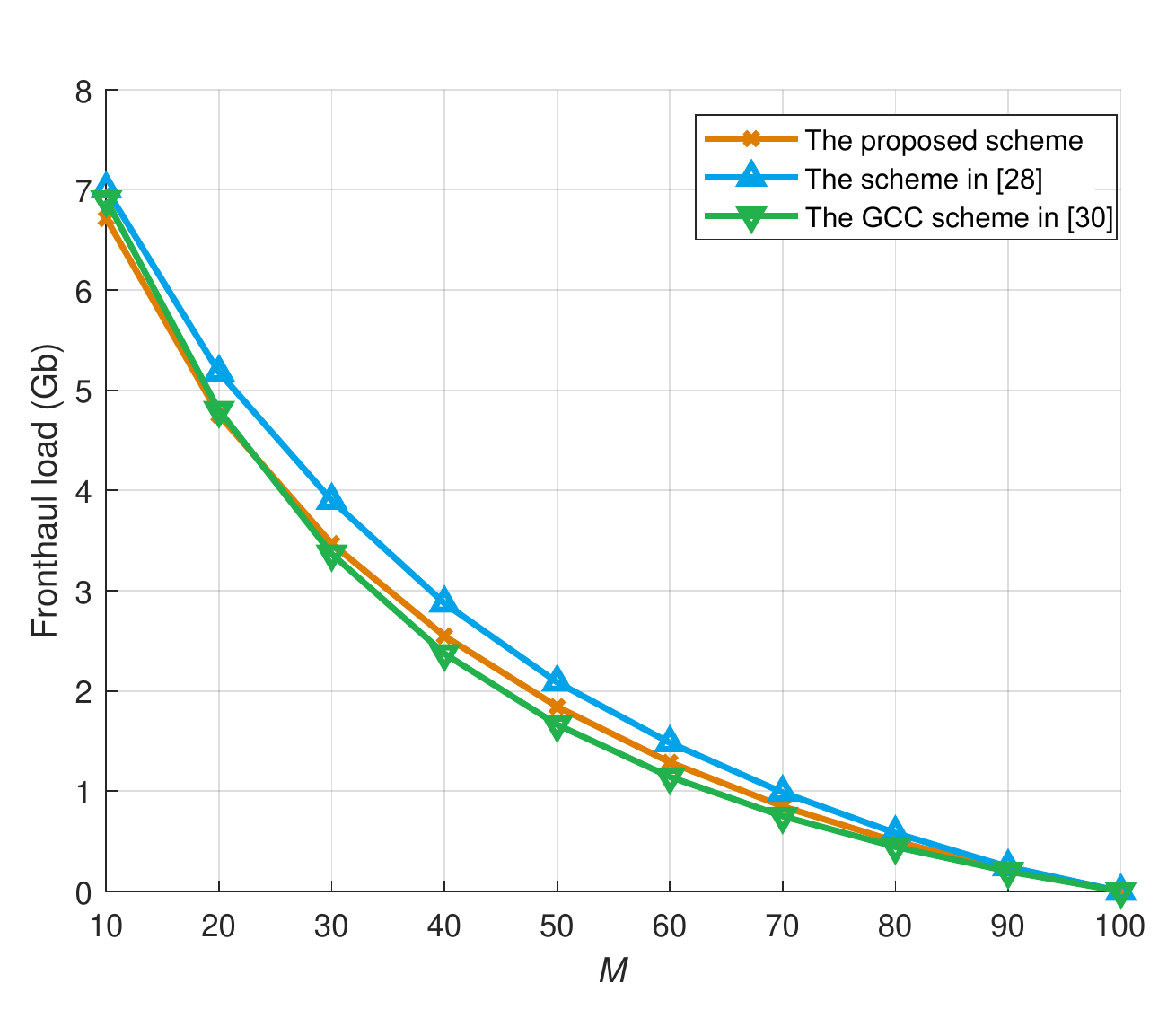}
\caption{Fronthaul load versus $M$ for different asynchronous coded caching schemes.} \label{fig-add}
\end{figure}

In Fig. \ref{user_load}, we show the effect of the number of F-APs requesting contents during each time slot  on the fronthaul load of each scheme with $M = 30$ for different $\Delta b$. As the request distribution affects the fronthaul load, we consider  the special case with $\left| {{{\cal U}_b}} \right| = L$.
As shown,
 the fronthaul load of our proposed scheme  increases more slowly with $L$ compared with that of the uncoded caching scheme.
 The reason for this result is that more coded-multicasting opportunities can be created when $L$ increases, i.e., $K$ increases.
 Correspondingly, when $K$ increases for $M \in \left[ {{N \mathord{\left/
 {\vphantom {N K}} \right.
 \kern-\nulldelimiterspace} K},N} \right]$, the maximum request delay can be set to a relatively smaller value while the same considerable coded-multicasting opportunities are created.

In Fig. \ref{cache_load}, we show the effect of the normalized cache size of each F-AP, i.e., $M$, on the fronthaul load of each scheme with $K = 10$ for different $\Delta b$.
Here we consider the general case with random $\left| {{{\cal U}_b}} \right|$.
 As shown, our  proposed  scheme can create considerable coded-multicasting opportunities compared with the uncoded caching scheme.
Moreover, the fronthaul load decreases and its slope increases when $M$ increases, which is the same as  the Maddah-Ali-Niesen's decentralized  scheme.
Furthermore, the fronthaul load of our proposed scheme is between the lower bound, i.e., the fronthaul load of the Maddah-Ali-Niesen's decentralized synchronous coded caching scheme, and the upper bound, and approaches the upper bound when $M$ increases.

In Fig. \ref{joint}, we show how $\Delta b$  affects the fronthaul load of  each scheme for varying cache sizes with random
 $\left| {{{\cal U}_b}} \right|$ and $K = 10$.
As shown,  the fronthaul load of our proposed scheme decreases with
$\Delta b$, which means that our proposed scheme can create more coded-multicasting opportunities with
a relaxed delay requirement.
Furthermore,  the larger $\Delta b$ is, the more the decrease of the fronthaul load of our proposed scheme is  in comparison with that of the uncoded caching scheme.
The performance gap between the fronthaul load of our proposed scheme and that of the Maddah-Ali-Niesen's decentralized
 scheme is smaller when $\Delta b$ is larger.
The reason for the above results is that a larger  $\Delta b$  leads to   a smaller number of  the partitioned subsets.
Besides,
as $\Delta b$ determines the maximum request delay,
it can be set to a relatively small value  in delay-sensitive scenarios and adjusted flexibly to achieve the  load-delay tradeoff in other scenarios.

In Fig. \ref{fig-add}, we show the performance comparison among our proposed scheme, the scheme in \cite{Niesen2}, and the GCC scheme in \cite{add8}.
In the simulations, the maximum request delay $\Delta b$ is set to $3$ for our proposed scheme and the threshold $\tau$ is set to $3$ for the scheme in \cite{Niesen2}.
As shown, the performance of our proposed scheme is better than that in \cite{Niesen2}.
The reason is that
the scheme in \cite{Niesen2} puts all requests into a sequence, and
a new request can be merged with the queued requests but may miss to merge with the upcoming requests for the $\tau$-fit threshold rule,
which is more appropriate to the situations where the sequence of requests is short, the file size is small, or the number of users is not large.
{
Although the performance gap between our proposed scheme and the scheme in \cite{Niesen2} with $\tau =3$ is small, it is not always the case.
The performance gap will become larger when the maximum request delay $\Delta b$, the threshold $\tau$, the queue length or the F-AP number is larger.
According to \cite{Niesen2}, it can be readily verified that our proposed scheme satisfies the perfect-fit rule with $\tau=0$, i.e., 0-fit threshold rule.
Moreover, a larger $\tau$ or queue length can indeed bring larger performance gap.
Besides, with a larger $\Delta b$ or F-AP number, the coded-multicasting opportunities among the considered F-APs will  increase.
Correspondingly, the scheme in \cite{Niesen2} will result in a larger probability to miss the coded-multicasting opportunities, and the performance gap will become larger.}
Furthermore, the GCC scheme  in \cite{add8} has a slight better performance than our proposed scheme when $M \ge 20$.
The reason is that the GCC scheme is centralized and the server knows the identity and the number of users exactly in the placement phase.
Correspondingly, it can be carefully designed to create more coded-multicasting opportunities.
However, the GCC scheme  in \cite{add8} cannot be applied to the networks with a variable number of users whereas both our proposed scheme and the scheme in  \cite{Niesen2} can be due to their decentralized property.
{
As for the complexity, in the delivery phase, the GCC scheme in \cite{add8} and our proposed scheme need some loops to encode the subfiles in each time slot, and the number of loops in a time slot can be calculated to be
   $\left\{\begin{array}{ll}
\mathcal{O}\left ( \binom{K}{\frac{KM}{N}+1}-\binom{K-\Delta bL}{\frac{KM}{N}+1} \right ), & \Delta bL\leq K\left ( 1-\frac{M}{N} \right )-1,\\
\mathcal{O}\left ( \binom{K}{\frac{KM}{N}+1}\right ), & \text{else}.
\end{array}\right.$
    and $\mathcal{O}(2^K-2^{K-L}) $, respectively.
    In comparison, the scheme in \cite{Niesen2}
    needs to traverse the sequence of requests multiple times, and the maximum number of traversals in a time slot can be calculated to be
    $\mathcal{O}\left(1/2 LF'(1-M/N)[2LF'(1-M/N)\Delta b+(L-1)]\right)$, where $F'$ denotes the number of partitioned subfiles for each file.
}

\section{Conclusions}

In this paper, we have proposed a decentralized asynchronous coded caching scheme for the online case in F-RANs where users asynchronously request contents with the maximum request delay.
Our proposed  scheme provides asynchronous and synchronous transmission methods to fulfill the  delay requirements of different practical scenarios.
The analytical results have shown that
the fronthaul load of our proposed scheme is at
most a constant factor larger than that of the Maddah-Ali-
Niesen's decentralized  scheme for a given maximum request
delay.
The simulation results have shown that more coded-multicasting opportunities can be created when the maximum request delay increases in asynchronous request scenarios.
For the future work,
we would like to explore asynchronous coded caching with a
nonuniform popularity distribution.

%


\bibliographystyle{IEEEtran}
\bibliography{codedcaching}

\begin{thebibliography}{10}
\providecommand{\url}[1]{#1}
\csname url@samestyle\endcsname
\providecommand{\newblock}{\relax}
\providecommand{\bibinfo}[2]{#2}
\providecommand{\BIBentrySTDinterwordspacing}{\spaceskip=0pt\relax}
\providecommand{\BIBentryALTinterwordstretchfactor}{4}
\providecommand{\BIBentryALTinterwordspacing}{\spaceskip=\fontdimen2\font plus
\BIBentryALTinterwordstretchfactor\fontdimen3\font minus
  \fontdimen4\font\relax}
\providecommand{\BIBforeignlanguage}[2]{{%
\expandafter\ifx\csname l@#1\endcsname\relax
\typeout{** WARNING: IEEEtran.bst: No hyphenation pattern has been}%
\typeout{** loaded for the language `#1'. Using the pattern for}%
\typeout{** the default language instead.}%
\else
\language=\csname l@#1\endcsname
\fi
#2}}
\providecommand{\BIBdecl}{\relax}
\BIBdecl

\bibitem{Bennis1}
G.~Lee, W.~Saad, and M.~Bennis, ``An online secretary framework for fog network
  formation with minimal latency,'' in \emph{2017 IEEE Int. Conf. Commun.
  (ICC)}, May 2017, pp. 1--6.

\bibitem{Bennis2}
M.~S. ElBamby, M.~Bennis, and W.~Saad, ``Proactive edge computing in
  latency-constrained fog networks,'' in \emph{the 26th European Conf. Netw.
  Commun. (EuCnC)}, June 2017, pp. 1--6.

\bibitem{Zhang}
K.~Zhang, Y.~Mao, S.~Leng, Q.~Zhao, L.~Li, X.~Peng, L.~Pan, S.~Maharjan, and
  Y.~Zhang, ``Energy-efficient offloading for mobile edge computing in {5G}
  heterogeneous networks,'' \emph{IEEE Access}, vol.~4, pp. 5896--5907, Aug.
  2016.

\bibitem{Bastug}
E.~Bastug, M.~Bennis, and M.~Debbah, ``Living on the edge: The role of
  proactive caching in {5G} wireless networks,'' \emph{IEEE Commun. Mag.},
  vol.~52, no.~8, pp. 82--89, Aug. 2014.

\bibitem{Wang}
X.~Wang, M.~Chen, T.~Taleb, A.~Ksentini, and V.~C.~M. Leung, ``Cache in the
  air: Exploiting content caching and delivery techniques for {5G} systems,''
  \emph{IEEE Commun. Mag.}, vol.~52, no.~2, pp. 131--139, Feb. 2014.

\bibitem{add10}
Y.~Jiang, M.~Ma, M.~Bennis, F.~Zheng, and X.~You, ``User preference learning
  based edge caching for fog radio access network,'' \emph{IEEE Trans. Commun.
  (Early Access)}, pp. 1--16, Nov. 2018.

\bibitem{Maddah-Ali1}
M.~A. Maddah-Ali and U.~Niesen, ``Fundamental limits of caching,'' \emph{IEEE
  Trans. Inf. Theory}, vol.~60, no.~5, pp. 2856--2867, May 2014.

\bibitem{Maddah-Ali2}
------, ``Decentralized coded caching attains order-optimal memory-rate
  tradeoff,'' \emph{IEEE/ACM Trans. Netw.}, vol.~23, no.~4, pp. 1029--1040,
  Aug. 2015.

\bibitem{shan}
K.~Shanmugam, M.~Ji, A.~M. Tulino, J.~Llorca, and A.~G. Dimakis.,
  ``Finite-length analysis of caching-aided coded multicasting,'' \emph{IEEE
  Trans. Inf. Theory}, vol.~62, no.~10, pp. 5524--5537, Oct. 2016.

\bibitem{Jin}
S.~Jin, Y.~Cui, H.~Liu, and G.~Caire, ``Order-optimal decentralized coded
  caching schemes with good performance in finite file size regime,'' in
  \emph{2016 IEEE Global Commun. Conf. (GLOBECOM)}, Dec. 2016, pp. 1--7.

\bibitem{Niesen1}
U.~Niesen and M.~A. Maddah-Ali, ``Coded caching with nonuniform demands,''
  \emph{IEEE Trans. Inf. Theory}, vol.~63, no.~2, pp. 1146--1158, Feb. 2017.

\bibitem{Hachem}
J.~Hachem, N.~Karamchandani, and S.~Diggavi, ``Multi-level coded caching,'' in
  \emph{2014 IEEE Int. Symp. Inf. Theory}, June 2014, pp. 56--60.

\bibitem{Ji1}
M.~Ji, A.~M. Tulino, J.~Llorca, and G.~Caire, ``Order-optimal rate of caching
  and coded multicasting with random demands,'' \emph{IEEE Trans. Inf. Theory},
  vol.~63, no.~6, pp. 3923--3949, June 2017.

\bibitem{JZhang1}
J.~Zhang, X.~Lin, and X.~Wang, ``Coded caching under arbitrary popularity
  distributions,'' \emph{IEEE Trans. Inf. Theory}, vol.~64, no.~1, pp.
  349--366, Jan. 2018.

\bibitem{Ji2}
M.~Ji, A.~M. Tulino, J.~Llorca, and G.~Caire, ``Caching and coded multicasting:
  Multiple groupcast index coding,'' in \emph{2014 IEEE Global Conf. Signal
  Inf. Process. (GlobalSIP)}, Dec. 2014, pp. 881--885.

\bibitem{Ji3}
M.~Ji, K.~Shanmugam, G.~Vettigli, J.~Llorca, A.~M. Tulino, and G.~Caire, ``An
  efficient multiple-groupcast coded multicasting scheme for finite fractional
  caching,'' in \emph{2015 IEEE Int. Conf. Commun. (ICC)}, June 2015, pp.
  3801--3806.

\bibitem{JZhang}
J.~Zhang, X.~Lin, C.~C. Wang, and X.~Wang, ``Coded caching for files with
  distinct file sizes,'' in \emph{2015 IEEE Int. Symp. Inf. Theory (ISIT)},
  June 2015, pp. 1686--1690.

\bibitem{Amiri}
M.~M. Amiri, Q.~Yang, and D.~Gündüz, ``Decentralized caching and coded
  delivery with distinct cache capacities,'' \emph{IEEE Trans. Commun.},
  vol.~65, no.~11, pp. 4657--4669, Nov. 2017.

\bibitem{Pedarsani}
R.~Pedarsani, M.~A. Maddah-Ali, and U.~Niesen, ``Online coded caching,''
  \emph{IEEE/ACM Trans. Netw.}, vol.~24, no.~2, pp. 836--845, Apr. 2016.

\bibitem{add1}
Q.~Yan, M.~Cheng, X.~Tang, and Q.~Chen, ``On the placement delivery array
  design for centralized coded caching scheme,'' \emph{IEEE Trans. Inf.
  Theory}, vol.~63, no.~5, pp. 5821--5833, Sep. 2017.

\bibitem{add2}
C.~Shangguan, Y.~Zhang, and G.~Ge, ``Centralized coded caching schemes: A
  hypergraph theoretical approach,'' \emph{IEEE Trans. Inf. Theory}, vol.~64,
  no.~8, pp. 5755--5766, Aug. 2018.

\bibitem{add3}
L.~Tang and A.~Ramamoorthy, ``Coded caching schemes with reduced
  subpacketization from linear block codes,'' \emph{IEEE Trans. Inf. Theory},
  vol.~64, no.~4, pp. 3099--3120, Apr. 2018.

\bibitem{add4}
H.~Ghasemi and A.~Ramamoorthy, ``Improved lower bounds for coded caching,''
  \emph{IEEE Trans. Inf. Theory}, vol.~63, no.~7, pp. 4388--4413, Jul. 2017.

\bibitem{add5}
C.~Tian and J.~Chen, ``Caching and delivery via interference elimination,''
  \emph{IEEE Trans. Inf. Theory}, vol.~64, no.~3, pp. 1548--1560, Mar. 2018.

\bibitem{add6}
J.~G{\'o}mez-Vilardeb{\'o}, ``A novel centralized coded caching scheme with
  coded prefetching,'' \emph{IEEE J. Sel. Areas Commun.}, vol.~36, no.~6, pp.
  1165--1175, Jun. 2018.

\bibitem{add7}
K.~Zhang and C.~Tian, ``Fundamental limits of coded caching: From uncoded
  prefetching to coded prefetching,'' \emph{IEEE J. Sel. Areas Commun.},
  vol.~63, no.~6, pp. 1153--1164, Jun. 2018.

\bibitem{Maddah-Ali3}
M.~A. Maddah-Ali and U.~Niesen, ``Coding for caching: Fundamental limits and
  practical challenges,'' \emph{IEEE Commun. Mag.}, vol.~54, no.~8, pp. 23--29,
  Aug. 2016.

\bibitem{Niesen2}
U.~Niesen and M.~A. Maddah-Ali, ``Coded caching for delay-sensitive content,''
  in \emph{2015 IEEE Int. Conf. Commun. (ICC)}, June 2015, pp. 5559--5564.

\bibitem{Ghasemi}
H.~Ghasemi and A.~Ramamoorthy, ``Asynchronous coded caching,'' in \emph{2017
  IEEE Int. Symp. Inf. Theory (ISIT)}, June 2017, pp. 2438--2442.

\bibitem{add8}
Y.~Lu, W.~Chen, and H.~V. Poor, ``Coded joint pushing and caching with
  asynchronous user requests,'' \emph{IEEE J. Sel. Areas Commun.}, vol.~36,
  no.~8, pp. 1843--1856, Aug. 2018.

\bibitem{add9}
W.~Huang, Y.~Jiang, M.~Bennis, F.~Zheng, H.~Gacanin, and X.~You,
  ``Decentralized asynchronous coded caching in fog-ran,'' in \emph{Proc. IEEE
  VTC 2018 Fall}, Chicago, USA, Aug. 2018, pp. 1--6.

\end{thebibliography}

\begin{biography}[{\includegraphics[width=1in,height
=1.25in,clip,keepaspectratio]{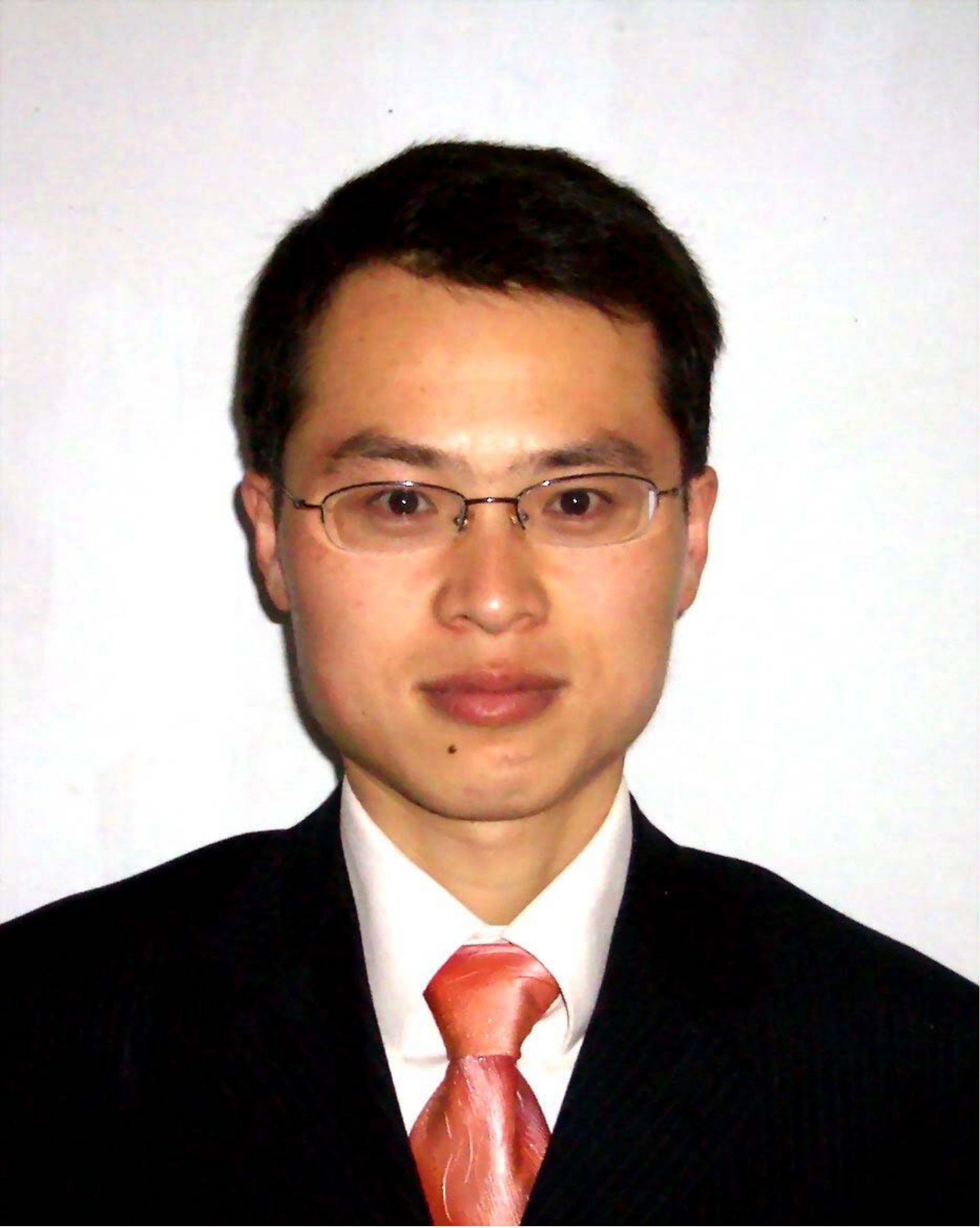}}]
{Yanxiang Jiang (S'03-M'07-SM'18)} received the B.S. degree in electrical engineering from Nanjing University, Nanjing, China, in 1999 and the M.S. and Ph.D. degrees in communications and information systems from Southeast University, Nanjing, China, in 2003 and 2007, respectively.

Dr. Jiang was a Visiting Scholar with the Signals and Information Group, Department of Electrical and Computer Engineering, University of Maryland at College Park, College Park, MD, USA, in 2014. He is currently an Associate Professor with the National Mobile Communications Research Laboratory, Southeast University, Nanjing, China. His research interests are in the area of broadband wireless mobile communications, covering topics such as edge caching, radio resource allocation and management, fog radio access networks, small cells and heterogeneous networks, cooperative communications, green communications, device to device communications, massive MIMO, and machine learning for wireless communications.
\end{biography}

\begin{biography}
{Wenlong Huang}
is currently pursuing the M.S. degree in communications and information systems from Southeast University, Nanjing, China.

His research interests include radio resource management and edge caching.
\end{biography}


\begin{biography}[{\includegraphics[width=1in,height
=1.25in,clip,keepaspectratio]{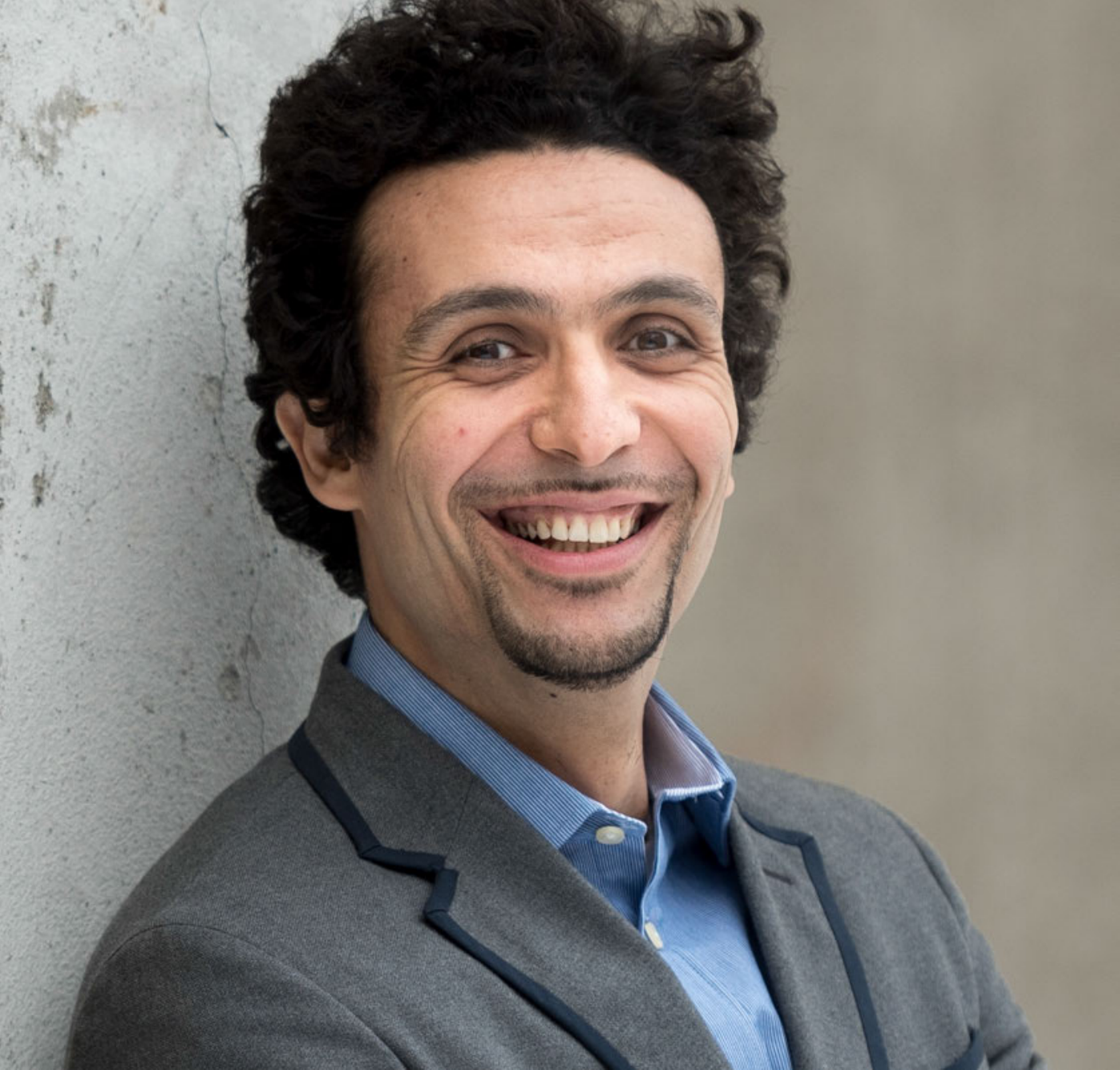}}]
{Mehdi Bennis (S'07-AM'08-SM'15)}
received his M.Sc. degree in electrical engineering jointly from EPFL, Switzerland, and the Eurecom Institute, France, in 2002. He obtained his Ph.D. from the University of Oulu in December 2009 on spectrum sharing for future mobile  cellular systems. Currently he is an associate professor at the  University of Oulu and an Academy of Finland research fellow.  His main research interests are in radio resource management, heterogeneous networks, game theory, and machine learning  in 5G networks and beyond. He has co-authored one book  and published more than 200 research papers in international  conferences, journals, and book chapters. He was the recipient  of the prestigious 2015 Fred W. Ellersick Prize from the IEEE  Communications Society, the 2016 Best Tutorial Prize from  the IEEE Communications Society, the 2017 EURASIP Best  Paper Award for the Journal of Wireless Communications and  Networks, and the 2017 all-University of Oulu Award for Research.
\end{biography}


\begin{biography}[{\includegraphics[width=1in,height
=1.25in,clip,keepaspectratio]{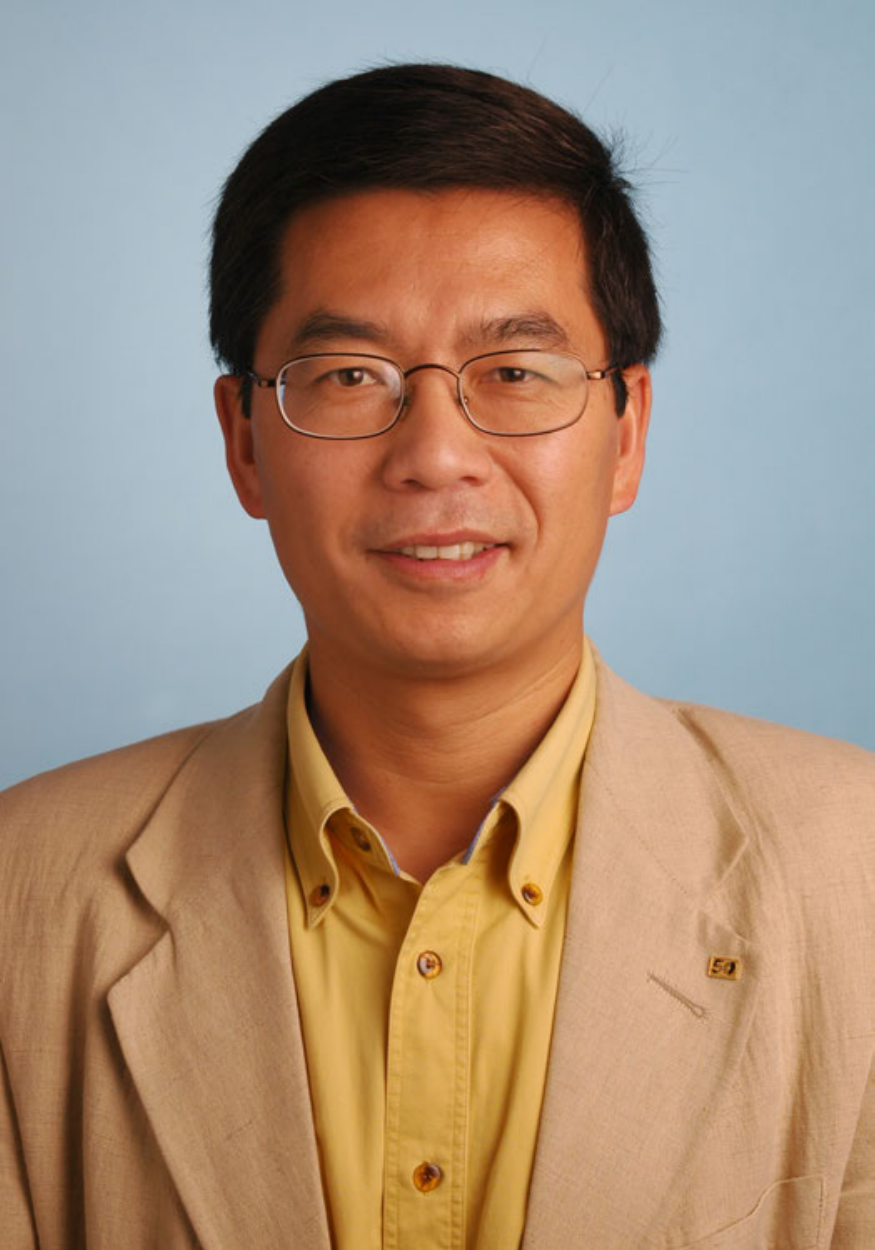}}]
{Fu-Chun Zheng (M'95-SM'99)}
obtained the BEng (1985) and MEng (1988) degrees in radio engineering from Harbin Institute of Technology, China, and the PhD degree in Electrical Engineering from the University of Edinburgh, UK, in 1992.

From 1992 to 1995, he was a post-doctoral research associate with the University of Bradford, UK, Between May 1995 and August 2007, he was with Victoria University, Melbourne, Australia, first as a lecturer and then as an associate professor in mobile communications.  He was with the University of Reading, UK, from September 2007 to July 2016 as a Professor (Chair) of Signal Processing. He has also been a distinguished adjunct professor with Southeast University, China, since 2010. Since August 2016, he has been a distinguished professor with Harbin Institute of Technology (Shenzhen), China and the University of York, UK. He has been awarded two UK EPSRC Visiting Fellowships - both hosted by the University of York (UK): first in August 2002 and then again in August 2006. Over the past two decades, Dr Zheng has also carried out many government and industry sponsored research projects - in Australia, the UK, and China. He has been both a short term visiting fellow and a long term visiting research fellow with British Telecom, UK. Dr Zheng's current research interests include signal processing for communications, multiple antenna systems, green communications, and ultra-dense networks.

He has been an active IEEE member since 1995. He was an editor (2001 - 2004) of IEEE Transactions on Wireless Communications. In 2006, Dr Zheng served as the general chair of IEEE VTC 2006-S, Melbourne, Australia (www.ieeevtc.org/vtc2006spring) - the first ever VTC held in the southern hemisphere in VTC's history of six decades. More recently he was the executive TPC Chair for VTC 2016-S, Nanjing, China (the first ever VTC held in mainland China: www.ieeevtc.org/vtc2016spring).

\end{biography}


\begin{appendices}
      \section{PROOF OF THEOREM 1  }

$R_A\left( {M,N,K,\Delta b} \right)$ can be calculated by accumulating the sizes of  all the  coded-multicasting contents transmitted by the cloud server.
Recall that there are $K$ types of encoding sets in total in the Maddah-Ali-Niesen's decentralized scheme, and there are
\begin{small}
$\left( {\begin{array}{*{20}{c}}
K\\
s
\end{array}} \right)$
\end{small}
encoding sets for  each type $s$.
Moreover, the size of the coded-multicasting content for any encoding set of type $s$ is identical, and different encoding sets of type $s$ can be partitioned into the same number of subsets.
   Note that the number of subsets that an encoding set of type $s$ can be  partitioned into ranges from ${\left\lceil {\frac{{s }}{{\Delta b \cdot L}}} \right\rceil }$ to ${\min \left\{ {\left\lceil {\frac{B}{{\Delta b}}} \right\rceil ,s } \right\}}$.
  Let $q\left( {s,Y,\Delta b } \right)$  denote the number of encoding sets of type $s$, each of which can be partitioned into $Y \in \left\{ {\left\lceil {\frac{{s }}{{\Delta b \cdot L}}} \right\rceil ,\left\lceil {\frac{{s }}{{\Delta b \cdot L}}} \right\rceil  + 1, \ldots ,\min \left\{ {\left\lceil {\frac{B}{{\Delta b}}} \right\rceil ,s } \right\}} \right\}$ subsets.
  Let $Q\left( s,\Delta b \right)$  denote the number of subsets that all the encoding sets of type $s$ is partitioned into.

  In the following,
  $b\left( {Y,\alpha } \right)$ and $c\left( {g,e} \right)$ are introduced firstly. Secondly, ${q_1}\left( {s,Y,{{\Delta b}^\prime},\Delta b } \right)$ and  ${q_2}\left( {s,Y,\Delta b } \right)$ are presented. Thirdly,
 for each type $s$, $q\left( {s,Y,\Delta b } \right)$ can be obtained for any $Y$ by using ${q_1}\left( {s,Y,{{\Delta b}^\prime},\Delta b } \right)$ and  ${q_2}\left( {s,Y,\Delta b } \right)$.
 Then,  $Q\left( s,\Delta b \right)$  can be obtained by summing $q\left( {s,Y,\Delta b} \right)  Y $ over all the possible values of $Y$. Finally, $R_A\left( {M,N,K,\Delta b} \right)$ can be calculated by summing ${Q\left( s,\Delta b \right)}  \left| {{W_{k,{
{{\cal S} \backslash \left\{ k \right\}}}}}} \right|$
 over all the possible values of $s$.

\subsection{ $b\left( {Y,\alpha } \right)$ and $c\left( {g,e} \right)$ }

\begin{enumerate}

\item

  We regard  $\Delta b$  consecutive time slots as a whole, which is called a big time slot.
  Let $Y$ encoding subsets correspond to  $Y$ big time slots, which are denoted by ${{{\cal B}}_1},{{{\cal B}}_1},\ldots,{{{\cal B}}_y},\ldots{{{\cal B}}_{Y}}$ and distributed during $B$ time slots. In other words, the F-APs in the $y$-th encoding subset request contents during big time slot ${{{\cal B}}_y}$.
  Let ${b_{{{\cal B}_y},1}}$ denote the index of the first time slot of big time slot ${{{\cal B}}_y}$.
  Accordingly, let ${\cal U}_{{b_{{{\cal B}_y},1}}}$ and ${\cal U}_{{{\cal B}}_y}$  denote the index set of the F-APs that request contents during time slot ${b_{{{\cal B}_y},1}}$ and big time slot ${{{\cal B}}_y}$, respectively. Then, ${{\cal U}_{{{\cal B}_y}}} = \mathop  \cup \nolimits_{b = {{b_{{{\cal B}_y},1}}}}^{{{b_{{{\cal B}_y},1}}} + \Delta b - 1} {{\cal U}_b}$.

Let $b\left( {Y,\alpha } \right)$ with $Y \le \alpha  \le Y  L$ denote the number of all the possible F-AP sets by choosing  $\alpha $  F-APs from all the F-APs that request contents  during   time slot ${b_{{{\cal B}_1},1}},{b_{{{\cal B}_2},1}},\ldots,{b_{{{\cal B}_y},1}},\ldots,{b_{{{\cal B}_Y},1}}$, i.e., the F-APs in $\mathop  \cup \nolimits_{y = 1}^Y {{\cal U}_{{{b_{{{\cal B}_y},1}}}}}$. Note that the number of F-APs in ${\cal U}_{{b_{{{\cal B}_y},1}}}$ that can be chosen ranges from $1$ to $\min \left\{ {L,\alpha - \left( {Y - 1} \right)} \right\}$, and at least one F-AP in ${{\cal U}_{{{b_{{{\cal B}_y},1}}}}}$ should be chosen according to our proposed encoding set partition method.
Then, $b\left( {Y,\alpha } \right)$ can be calculated  by considering the following four  cases.

\begin{itemize}
\item

When $Y = 1$ and  $1 \le \alpha  \le   L$, choose $\alpha$ F-APs from the $L$ F-APs. Then,
we have
\begin{equation} \label{b2}
b\left( {Y,\alpha } \right) = {\left( {\begin{array}{*{20}{c}}
L\\
\alpha
\end{array}} \right)}, \quad {Y = 1,1 \le \alpha  \le L}.
\end{equation}

\item
When $Y > 1$ and ${\alpha  = Y}$, it can be readily seen that only one F-AP can be chosen from ${\cal U}_{{b_{{{\cal B}_y},1}}}$, and the number of possible results is  $\begin{small}
\left( {\begin{array}{*{20}{c}}
L\\
1
\end{array}} \right)
\end{small}$. Since choosing an F-AP during each time slot is independent with each other, we have
\begin{equation} \label{b3}
b\left( {Y,\alpha } \right) = {\left( {\begin{array}{*{20}{c}}
L\\
1
\end{array}} \right)^Y}, \quad {Y > 1,\alpha  = Y}.
\end{equation}

\item
When $Y >1$, $L > 1$, and ${\alpha  = Y  L}$, it means that the corresponding $L$ F-APs are chosen from ${\cal U}_{{b_{{{\cal B}_y},1}}}$ since there are $Y  L$ F-APs requesting contents during the $Y$ time slots in total.
 Then, we have
\begin{equation} \label{b4}
b\left( {Y,\alpha } \right) = 1, \quad {Y > 1,L > 1,\alpha  = YL}.
\end{equation}

\item
Otherwise,
note that the number of F-APs that are chosen from ${\cal U}_{{b_{{{\cal B}_1},1}}}$, denoted by $v$,  ranges from  $1$ to $\min \left\{ {L,\alpha  - \left( {Y - 1} \right)} \right\}$. When $ v = 1$, the number of all the possible results by choosing $\alpha - 1$ F-APs from $\mathop  \cup \nolimits_{y = 2}^Y {{\cal U}_{{{b_{{{\cal B}_y},1}}}}}$ is $b\left( {Y - 1,\alpha  - 1} \right)$. Meanwhile,  the number of all the possible results by choosing one F-AP from ${\cal U}_{{b_{{{\cal B}_1},1}}}$ is
\begin{small}
$\left( {\begin{array}{*{20}{c}}
L\\
1
\end{array}} \right)$
\end{small}.
Since  choosing F-APs from ${\cal U}_{{b_{{{\cal B}_y},1}}}$ during  time slot ${b_{{{\cal B}_y},1}}$ is independent with each other, the number of F-AP sets with $ v = 1$ is
\begin{equation}
\left( {\begin{array}{*{20}{c}}
L\\
1
\end{array}} \right)b\left( {Y - 1,\alpha  - 1} \right).
\end{equation}

    Repeat the above operations until  $v = \min \left\{ {L,\alpha - \left( {Y - 1} \right)} \right\}$. Sum  the number of  F-AP sets over all the possible values of $v$. Then,  we have
    \begin{equation} \label{b1}
b\left( {Y,\alpha } \right) = \sum\limits_{v = 1}^{\min \left\{ {L,\alpha  - \left( {Y - 1} \right)} \right\}} {\left( {\begin{array}{*{20}{c}}
L\\
v
\end{array}} \right)  b\left( {Y - 1,\alpha  - v} \right)}, \quad \rm{else.}
\end{equation}

\end{itemize}

According to (\ref{b2})-(\ref{b1}), we can readily obtain (\ref{b}).

\item
Let  $c\left( {g,e} \right)$ denote the number of all the possible results when  $g$  big time slots are distributed into $e$ placement, each of which represents a kind of  partition.
 Calculating $c\left( {g,e} \right)$ is equivalent to distributing $g$ indistinguishable balls into $e$ distinguishable boxes, which belongs to the issue in combinatorial mathematics.
 Since there is no difference among the big time slots and the boxes can be empty, we have\begin{equation} \label{c}
c\left( {g,e} \right) = \left( {\begin{array}{*{20}{c}}
{g + e - 1}\\
{g}
\end{array}} \right).
\end{equation}

\end{enumerate}

\subsection{ ${q_1}\left( {s,Y,{{\Delta b}^\prime},\Delta b } \right)$ and  ${q_2}\left( {s,Y,\Delta b } \right)$}

Before calculating $q\left( {s,Y,\Delta b } \right)$ with ${\Delta b < B}$, two possible partition cases with ${\Delta b < B}$, as shown in Fig. \ref{partition_case}, needs to be considered.
Note that big time slot ${\cal B}_Y$ includes ${{\Delta b}^\prime} < \Delta b $  time slots  and $\Delta b $  time slots in  case 1 and case 2, respectively.
 Let ${q_1}\left( {s,Y,{{\Delta b}^\prime},\Delta b } \right)$  and ${q_2}\left( {s,Y,\Delta b } \right)$ denote the number of all the encoding sets whose partition results correspond to case $1$ and case $2$, respectively.
Let $d_1\left( Y,{{\Delta b}^\prime},\Delta b  \right)$ and $d_2\left( Y,\Delta b  \right)$  denote the number of all the possible results by distributing $Y $ big time slots into $ B$ time slots for case 1 and case 2, respectively.
Let  ${p_1}\left( {s,Y,{{\Delta b}^\prime},\Delta b} \right)$ and ${p_2}\left( {s,Y,\Delta b} \right)$ denote the number of all the possible encoding sets of type $s$ by choosing $s$ F-APs from the F-APs in $\mathop  \cup \nolimits_{y = 1}^Y {{\cal U}_{{{\cal B}_y}}}$
for case 1 and case 2, respectively.
 Then, ${q_1}\left( {s,Y,{{\Delta b}^\prime},\Delta b } \right)$  and ${q_2}\left( {s,Y,\Delta b } \right)$ can be calculated as follows.

\subsubsection{${q_1}\left( {s,Y,{{\Delta b}^\prime},\Delta b } \right)$}

 As  the calculations of $d_1\left( Y,{{\Delta b}^\prime},\Delta b  \right)$ and  $p_1\left( {s,Y,{{\Delta b}^\prime},\Delta b } \right)$ are  independent with each other, we have
        \begin{equation} \label{q1}
{q_1}\left( {s,Y,{{\Delta b}^\prime} ,\Delta b} \right) = d_1\left( {Y,{{\Delta b}^\prime},\Delta b } \right){p_1}\left( {s,Y,{{\Delta b}^\prime},\Delta b } \right).
\end{equation}

As for $d_1\left( Y,{{\Delta b}^\prime},\Delta b  \right)$, it is equal to the number of all the results by distributing $Y-1$ big time slots into $B - \left( {Y  - 1} \right)\Delta b - {{\Delta b}^\prime} + 1$ placements.
    According to (\ref{c}), we have
    \begin{equation} \label{d1}
\begin{split}
 d_1\left( Y,{{\Delta b}^\prime},\Delta b  \right)
   &  =   {c\left( {Y  - 1,B - \left( {Y  - 1} \right)\Delta b - {{\Delta b}^\prime} + 1} \right)}  \\
 &  =  {\left( {\begin{array}{*{20}{c}}
{B - {{\Delta b}^\prime} - \left( {Y  - 1} \right)\left( {\Delta b - 1} \right)}\\
{Y-1}
\end{array}} \right)}.
 \end{split}
\end{equation}

As for ${p_1}\left( {s,Y,{{\Delta b}^\prime},\Delta b } \right)$, it can be calculated in two steps. The first step is to calculate  $b\left( {Y ,\alpha }\right)$ with
\begin{multline}
\alpha  \in \left\{ \max \left\{ {Y,s  - \left( {\left( {Y - 1} \right)\Delta b + {{\Delta b}^\prime}   - Y} \right)L} \right\}, \right. \\
\left. \ldots,\min \left\{ {s ,YL} \right\} \right\}.
\end{multline}
The second step is to calculate the number of F-AP sets by choosing $s-\alpha$ F-APs that request contents during the remaining time slots of $Y$ big time slots, i.e.,
\begin{small}
${\left( {\begin{array}{*{20}{c}}
{\left( {\left( {Y - 1} \right)\Delta b + {{\Delta b}^\prime}  - Y} \right)L}\\
{s - \alpha  }
\end{array}} \right)}$.
\end{small}
Sum
 \begin{small}
 ${b\left( {Y ,\alpha } \right)\left( {\begin{array}{*{20}{c}}
{\left( {\left( {Y - 1} \right)\Delta b + {{\Delta b}^\prime}  - Y} \right)L}\\
{s - \alpha  }
\end{array}} \right)}$
\end{small}
 over all the possible values of $\alpha$. Then, we have
\begin{small}
    \begin{equation}\label{p1}
    \begin{split}
&{p_1}\left( {s,Y,{{\Delta b}^\prime},\Delta b } \right) \\
&= \sum\limits_{\alpha  = \max \left\{ {Y,s  - \left( {\left( {Y - 1} \right)\Delta b + {{\Delta b}^\prime}  - Y} \right)L} \right\}}^{\min \left\{ {s ,YL} \right\}} {b\left( {Y,\alpha } \right)\left( {\begin{array}{*{20}{c}}
{\left( {\left( {Y - 1} \right)\Delta b + {{\Delta b}^\prime}  - Y} \right)L}\\
{s - \alpha  }
\end{array}} \right)} .
\end{split}
\end{equation}
\end{small}

According to (\ref{q1})-(\ref{d1}) and (\ref{p1}), we can readily obtain (\ref{q1Theo1}).

 \begin{figure}[!t]
\centering
\subfigure[Case 1]{
\label{case_a} 
\includegraphics[width=3in]{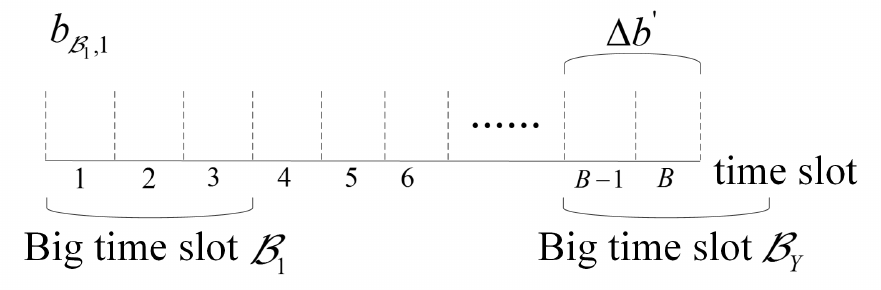}}
\hspace{1in}
\subfigure[Case 2]{
\label{case_b} 
\includegraphics[width=3in]{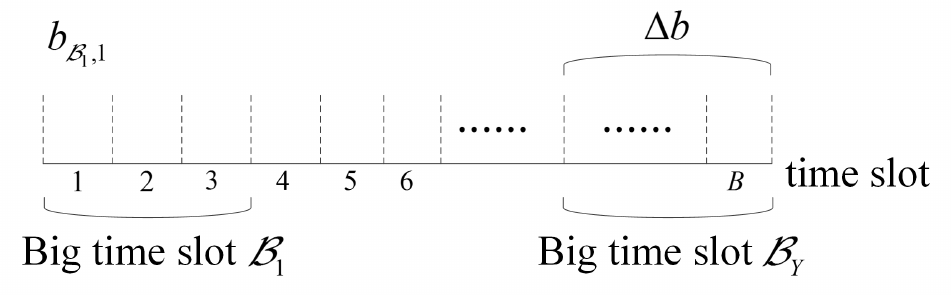}}
\caption{Two partition cases with ${\Delta b < B}$.}
\label{partition_case} 
\end{figure}

\subsubsection{${q_2}\left( {s,Y,\Delta b } \right)$ }

Similarly, we have
\begin{equation} \label{q2d2p2}
{q_2}\left( {s,Y,\Delta b} \right) = {d_2}\left( {Y,\Delta b} \right){p_2}\left( {s,Y,\Delta b} \right),
\end{equation}
 which can be calculated by  considering the following two cases, i.e., ${\Delta b = 1}$ and $1 < \Delta b < B$.
\begin{itemize}
\item

When ${\Delta b = 1}$, each big time slot only has one time slot. Then, we have
\begin{equation} \label{d21}
\begin{split}
 d_2\left( Y,\Delta b  \right)
   &     = c\left( {Y,B - Y  + 1 } \right) \\
 &    = \left( {\begin{array}{*{20}{c}}
B\\
Y
\end{array}} \right).
 \end{split}
\end{equation}
Similarly, we have
\begin{equation} \label{p21}
p_2\left( {s,Y,\Delta b } \right) = {b\left( {Y ,s } \right)}.
\end{equation}
According to (\ref{q2d2p2})-(\ref{p21}), we have
\begin{equation}  \label{q21}
{q_2}\left( {s,Y,\Delta b} \right) = \left( {\begin{array}{*{20}{c}}
B\\
Y
\end{array}} \right)b\left( {Y,s} \right), \quad {\Delta b = 1}.
\end{equation}

\item
When $1 < \Delta b < B$, $d_2\left( Y,\Delta b  \right)$ is equal to the number of all the possible results by distributing $Y$ big time slots into $B - Y\Delta b  + 1$ placements.
Similarly, we have
\begin{equation} \label{d2}
{d_2}\left( {Y,\Delta b} \right) = c\left( {Y,B - Y\Delta b + 1} \right) = \left( {\begin{array}{*{20}{c}}
{B - Y\left( {\Delta b - 1} \right)}\\
Y
\end{array}} \right),
\end{equation}
\begin{multline} \label{p2}
{p_2}\left( {s,Y,\Delta b} \right) = \\
\sum\limits_{\alpha  = \max \left\{ {Y,s  - Y\left( {\Delta b - 1} \right)L} \right\}}^{\min \left\{ {s, Y  L} \right\}} {b\left( {Y,\alpha } \right)\left( {\begin{array}{*{20}{c}}
{Y\left( {\Delta b - 1} \right)L}\\
{s  - \alpha }
\end{array}} \right)}.
\end{multline}

According to (\ref{q2d2p2}) and (\ref{d2})-(\ref{p2}), we have
\begin{small}
 \begin{equation} \label{q22}
\begin{array}{l}
{q_2}\left( {s,Y,\Delta b} \right) = \left( {\begin{array}{*{20}{c}}
{B - Y\left( {\Delta b - 1} \right)}\\
Y
\end{array}} \right)\\
\sum\limits_{\alpha  = \max \left\{ {Y,s  - Y\left( {\Delta b - 1} \right)L} \right\}}^{\min \left\{ {s,Y  L} \right\}} {b\left( {Y,\alpha } \right)\left( {\begin{array}{*{20}{c}}
{Y\left( {\Delta b - 1} \right)L}\\
{s  - \alpha }
\end{array}} \right)}
\end{array},  \quad 1 < \Delta b < B.
\end{equation}
\end{small}

\end{itemize}

Finally, According to (\ref{q21}) and (\ref{q22}), we can readily obtain (\ref{q2}).

\subsection{ $q\left( {s,Y,\Delta b } \right)$ }

When ${\Delta b = B}$,
according to our proposed encoding set partition method, any encoding set will be partitioned into one subset. As the number of all the encoding sets of type $s$ is
\begin{small}
${\left( {\begin{array}{*{20}{c}}
K\\
s
\end{array}} \right)}$
\end{small},  we have
\begin{equation} \label{qB}
q\left( {s,Y,\Delta b } \right) = {\left( {\begin{array}{*{20}{c}}
K\\
s
\end{array}} \right)}, \quad {\Delta b = B}.
\end{equation}

When ${\Delta b < B}$, $q\left( {s,Y,\Delta b } \right)$  can be calculated by considering the following four cases. Note that the number of time slots that big time slot ${\cal B}_Y$ includes is at least $\max \left\{ {\left\lceil {\frac{{s  - \left( {Y - 1} \right)L}}{L}} \right\rceil ,1} \right\}$ and at most $\min \left\{ {B - \left( {Y - 1} \right)\Delta b,\Delta b} \right\}$.
 \begin{itemize}
\item
When  ${\Delta b = 1}$ or
$\Delta b = \max \left\{ {\left\lceil {\frac{{s  - \left( {Y - 1} \right)L}}{L}} \right\rceil ,1} \right\} < B$, $q\left( {s,Y,\Delta b } \right)$  only includes the number of encoding sets of type $s$ whose partition results  correspond to case $2$. Then, we have
\begin{small}
     \begin{equation} \label{qb1}
q\left( {s,Y,\Delta b} \right) = \left\{ {\begin{array}{*{20}{c}}
{{q_2}\left( {s,Y,\Delta b} \right),}&{\Delta b = 1,}\\
{{q_2}\left( {s,Y,\Delta b} \right),}&{\Delta b = \max \left\{ {\left\lceil {\frac{{s  - \left( {Y - 1} \right)L}}{L}} \right\rceil ,1} \right\} < B.}
\end{array}} \right.
\end{equation}
\end{small}

\item  Similarly, when
$\max \left\{ {\left\lceil {\frac{{s  - \left( {Y - 1} \right)L}}{L}} \right\rceil ,1} \right\} < \Delta b \le B - \left( {Y - 1} \right)\Delta b$,
     both  case $1$ and case $2$ are included.

        As
        ${{\Delta b}^\prime}  \in \left\{ {\max \left\{ {\left\lceil {\frac{{s  - \left( {Y - 1} \right)L}}{L}} \right\rceil ,1} \right\},\ldots,{\Delta b} - 1} \right\}$
         for case 1,
      we have
      \begin{small}
      \begin{equation} \label{qb2}
      \begin{split}
q\left( {s,Y,\Delta b} \right) & = \sum\limits_{{{\Delta b}^\prime}  = \max \left\{ {\left\lceil {\frac{{s - \left( {Y - 1} \right)L}}{L}} \right\rceil ,1} \right\}}^{\Delta b - 1} {{q_1}\left( {s,Y,{{\Delta b}^\prime} ,\Delta b} \right)}  + {q_2}\left( {s,Y,\Delta b} \right),\\
&  \max \left\{ {\left\lceil {\frac{{s  - \left( {Y - 1} \right)L}}{L}} \right\rceil ,1} \right\} < \Delta b \le B - \left( {Y - 1} \right)\Delta b.
\end{split}
\end{equation}
\end{small}

\item When
$B - \left( {Y - 1} \right)\Delta b < \Delta b < B$, only case $1$ with ${{\Delta b}^\prime}  \in \left\{ {\max \left\{ {\left\lceil {\frac{{s  - \left( {Y - 1} \right)L}}{L}} \right\rceil ,1} \right\},\ldots,B - \left( {Y - 1} \right)\Delta b} \right\}$  is included.  Then, we have
    \begin{equation} \label{qb3}
    \begin{split}
q\left( {s,Y,\Delta b} \right) &= \sum\limits_{{{\Delta b}^\prime}  = \max \left\{ {\left\lceil {\frac{{s  - \left( {Y - 1} \right)L}}{L}} \right\rceil ,1} \right\}}^{B - \left( {Y - 1} \right)\Delta b} {{q_1}\left( {s,Y,{{\Delta b}^\prime} ,\Delta b} \right)}, \\
& \quad B - \left( {Y - 1} \right)\Delta b < \Delta b < B.
\end{split}
\end{equation}

\end{itemize}

According to (\ref{qB})-(\ref{qb3}), we  can readily obtain (\ref{q}).

\subsection{ $Q\left( s,\Delta b \right)$}

Recall that $Y $ is at least $\left\lceil {\frac{s}{{\Delta b \cdot L}}} \right\rceil  $  and at most $\min \left\{ {\left\lceil {B/\Delta b} \right\rceil ,s} \right\}$ for an encoding set of type $s$. Sum ${q\left( {s,Y,\Delta b} \right)}Y$ over all  the possible values of $Y $. Then, we have
 \begin{equation} \label{Q}
 Q\left( s,\Delta b \right) = \sum\limits_{Y = \left\lceil {\frac{{s }}{{\Delta b \cdot L}}} \right\rceil }^{\min \left\{ {\left\lceil {\frac{B}{{\Delta b}}} \right\rceil ,s } \right\}} {q\left( {s,Y,\Delta b} \right)} Y.
 \end{equation}

\subsection{ $R_A\left( {M,N,K,\Delta b} \right)$ with $\left| {{{\cal U}_b}} \right| = L$ and $  B \ge 3$ }

Finally,
sum  ${Q\left( s,\Delta b \right)} \left| {{W_{k,
{{\cal S} \backslash \left\{ k \right\}}}^{\rm a}}}\right|$ over all  the possible values of $s$. Then, we have
\begin{equation} \label{r1}
R_A\left( {M,N,K,\Delta b} \right) = \sum\limits_{s = 1}^{K} {Q\left( s,\Delta b \right)} \left| {{W_{k,
{{\cal S} \backslash \left\{ k \right\}}}^{\rm a}}} \right|, \quad \left| {{{\cal U}_b}} \right| = L,  B \ge 3.
\end{equation}

According to (\ref{Q}) and (\ref{r1}), we  can readily obtain (\ref{theo1}).

This completes the proof.

\section{PROOF OF THEOREM 2}

Accumulate the sizes of the coded-multicasting contents corresponding to all the subsets that all the encoding sets are partitioned into. Then, we have
\begin{equation} \label{4}
{R_A}\left( {M,N,K,\Delta b} \right) = F\sum\limits_{s = 1}^K {\sum\limits_{i = 1}^{\left( {\begin{array}{*{20}{c}}
K\\
s
\end{array}} \right)} {{\eta _{{\cal S}_i^\prime }}\left( {\Delta b} \right)\left| {W_{k,{\cal S}_i^\prime \backslash \left\{ k \right\}}^{\rm{a}}} \right|} } ,
\end{equation}
where ${{\cal S}^\prime_i}$ denotes the $i$-th element  of  the ${\left( {\begin{array}{*{20}{c}}
K\\
s
\end{array}} \right)}$ encoding sets of type $s $.

According to (\ref{partM}), (\ref{syn_load}), and (\ref{4}), we have
\begin{align}\label{AS}
 R_A\left( {M,N,K,\Delta b} \right) &\ge F\sum\limits_{s = 1}^K {\left( {\begin{array}{*{20}{c}}
K\\
s
\end{array}} \right){{\left( {M/N} \right)}^{s - 1}}{{\left( {1 - M/N} \right)}^{K - \left( {s - 1} \right)}}} \nonumber \\
&= {R_S}\left( {M,N,K} \right),
\end{align}

\begin{small}
\begin{align}\label{5}
 R_A\left( {M,N,K,\Delta b} \right) &\le F\sum\limits_{s = 1}^K {\left( {\begin{array}{*{20}{c}}
K\\
s
\end{array}} \right)\left\lceil {B/\Delta b} \right\rceil {{\left( {M/N} \right)}^{s - 1}}{{\left( {1 - M/N} \right)}^{K - \left( {s - 1} \right)}}} \nonumber \\
&= \left\lceil {B/\Delta b} \right\rceil {R_S}\left( {M,N,K} \right).
\end{align}
\end{small}

Let ${R_U}\left( M,N,K \right)$ denote the fronthaul load of the uncoded caching scheme. Then,  from \cite{Maddah-Ali2}, we have
\begin{equation}\label{6}
{R_U}\left( M,N,K \right) = F  K  \left( {1 - {M \mathord{\left/
 {\vphantom {M N}} \right.
 \kern-\nulldelimiterspace} N}} \right) , \quad N \ge K.
\end{equation}

As $R_A\left( {M,N,K,\Delta b} \right)$ increases  with $\Delta b$, we have
\begin{equation}\label{U1}
{R_A}\left( {M,N,K,\Delta b} \right) \le {R_A}\left( {M,N,K,1} \right).
\end{equation}
When $\Delta b = 1$, our proposed scheme is equivalent to the scheme applying the  Maddah-Ali-Niesen's decentralized scheme during each time slot.  From \cite{Maddah-Ali2}, we know that the fronthaul load of the  Maddah-Ali-Niesen's decentralized scheme is no more than that of the uncoded caching scheme  during each time slot.
\footnote{
Specifically, when $K = 1$, we have
${R_S}\left( {M,N,K} \right) = {R_U}\left( {M,N,K} \right).$
}
Then, accumulating the fronthaul load during $B$ time slots, we have
\begin{equation}\label{U2}
{R_A}\left( {M,N,K,1} \right) \le {R_U}\left( {M,N,K} \right).
\end{equation}
According to (\ref{U1}) and (\ref{U2}), we have
\begin{equation}
R_A\left( {M,N,K,\Delta b} \right) \le {R_U}\left( M,N,K \right). \label{7}
\end{equation}
According to (\ref{syn_load}), (\ref{5})-(\ref{6}), and (\ref{7}),  we have
\begin{small}
\begin{equation}
\begin{split}
R_A\left( {M,N,K,\Delta b} \right)
   &  \le \min \left\{ {\left\lceil {B/\Delta b} \right\rceil  \cdot {R_S}\left( M,N,K \right),{R_U}\left( M,N,K \right)} \right\}       \\
   &  = FK(1 - \frac{M}{N})\min \left\{ {\left\lceil {\frac{B}{{\Delta b}}} \right\rceil \frac{N}{{KM}}\left( {1 - {{\left( {1 - {M \mathord{\left/
 {\vphantom {M N}} \right.
 \kern-\nulldelimiterspace} N}} \right)}^K}} \right),1} \right\}. \label{8}
 \end{split}
\end{equation}
\end{small}

According to (\ref{AS}) and (\ref{8}), we can readily obtain (\ref{theo2}).

This completes the proof.

  \end{appendices}

\end{document}